%% file: 1TON.tex
\documentclass[journal]{IEEEtran}
\makeatletter
\def\ps@headings{%
\def\@oddhead{\mbox{}\scriptsize\rightmark \hfil \thepage}%
\def\@evenhead{\scriptsize\thepage \hfil \leftmark\mbox{}}%
\def\@oddfoot{}%
\def\@evenfoot{}}
\makeatother
\usepackage{latexsym,bm,amsmath,amssymb}
\usepackage[numbers]{natbib}
\usepackage{graphicx}
\usepackage{bm}
\usepackage{ifpdf}

\usepackage[english]{babel} 
\usepackage{amsthm}
\usepackage{array}
\usepackage{mdwmath}
\usepackage{mdwtab}
\usepackage[tight,footnotesize]{subfigure}
\usepackage[font=footnotesize]{subfig}

\newcommand{\nn}{\nonumber \\}
\allowdisplaybreaks
\newtheorem{thm}{Theorem}%[section]
\newtheorem{cor}{Corollary}%[section]
\IEEEoverridecommandlockouts

\begin{document}

\title{\textbf{Online Learning with Randomized Feedback Graphs for Optimal PUE Attacks in Cognitive Radio Networks }}
               
%              Optimal Online Learning with Randomized Feedback Graphs with Application in PUE Attacks in CRN }}
% Randomized Feedback Graph-based Optimal Online Learning with Application in PUE Attacks in CRN
%Optimal Online Learning based on Randomized Feedback Graphs with Application in PUE Attacks in CRN
% Randomized Time-Variable Feedback Graphs for Optimal Online Learning in PUE Attacks in CRN
%Optimal online learning with Randomized Time-Variable Feedback Graphs With Application in CRN Security
%Optimal Online Learning-Based PUE Attacks in CRN: Randomized Time-Variable Feedback Graphs
%Optimal Online Learning with Application in PUE Attacks in CRN: Randomized Feedback Graphs
\author{
\IEEEauthorblockN{Monireh Dabaghchian, Amir Alipour-Fanid, Kai Zeng, Qingsi Wang, Peter Auer}\\

\thanks {Kai Zeng is the corresponding author}
}
\maketitle

\begin{abstract}
In a cognitive radio network, a secondary user learns the spectrum environment and dynamically accesses the channel where the primary user is inactive.
%(CR)
At the same time, a primary user emulation (PUE) attacker can send falsified primary user signals and prevent the secondary user from utilizing the available channel.
%Although there is a large body of work on PUE attack detection and defending strategies, 
The best attacking strategies that an attacker can apply have not been well studied.
In this paper, for the first time, we study optimal PUE attack strategies 
%without any prior knowledge on the primary user activity characteristics and the secondary user access strategies.
by formulating an online learning problem where the attacker needs to dynamically decide the attacking channel in each time slot based on its attacking experience.
The challenge in our problem is that since the PUE attack happens in the spectrum sensing phase, the attacker cannot observe the 
% actual activity of the secondary user 
reward on the attacked channel.
% and it never knows if a secondary user ever tries to access it.
To address this challenge, we utilize the attacker's observation capability.
We propose online learning-based attacking strategies based on the attacker's observation capabilities. 
%  new attack schemes in which the attacker can either observe one or multiple channels in non-attacked slots or within the attacking slots.
%We propose two attacking schemes called attack-or-observe and attack-but-observe-another that capture the attacker's observation capability.
% (AORO) 
% (ABOA) 
%when the time-slot duration is short and large, respectively.or it can be due to the number of the antennas the attacker has.
%Based on the first scheme, the attacker either attacks or observes at each time.
%While, based on the second scheme, it both attacks and observes at least one other channel during the data transmission phase at each time.
%We propose online learning-based attacking strategies based on the attacker's observation capabilities.
%We propose multi-armed bandit and radomized feedback graphs techniques to make optimal online decisions on which channel(s) to attack or observe in each time slot to minimize the regret.
%We propose two online learning-based attacking strategies called POLA and OPT-RO based on the proposed attacking schemes, respectively. 
%POLA decides between observation and attack at each time dynamically.
%For observation, both POLA and OPT-RO choose a channel uniform randomly.
Through our analysis, we show that % the attacker's observation capability has the key impact on it's performance.
with no observation within the attacking slot, the attacker loses on the regret order, and with the observation of at least one channel, there is a significant improvement on the attacking performance.
Observation of multiple channels does not give additional benefit to the attacker (only a constant scaling) though it gives insight on the number of observations required to achieve the minimum constant factor. %compared to best algorithm.
%, and the change from null to something matters.
%More specifically, POLA and OPT-RO achieve a regret in the order of $\tilde{O}(\sqrt[3]{T^2})$ and $\tilde{O}(\sqrt{T})$, respectively.
%$T$ is the number of time slots the CR network operates.
Our proposed algorithms are optimal in the sense that their regret upper bounds match their corresponding regret lower-bounds.
%We also generalize OPT-RO to multichannel observation cases meaning the attacker has much more observation capability.
%The analysis leads to more insights on the attacker's observation capability impact on its performance.
We show consistency between simulation and analytical results under various system parameters.
\end{abstract}

\IEEEpeerreviewmaketitle

\input{2Introduction}
\input{3RelatedWork}

\input{4Problemformulation}

\input{5OptAttckStrtgy}
\input{6OptmlLearningPOLA}

\input{7RegretUbPOLA}

\input{8OptmlLearningPROLA}

\input{9RegretUbPROLA}
\input{10MultiAcExtensionPROLA}

\input{11Simulation}
\input{12Simltion1}
\input{13Simltion2}
\input{14Simltion3}
\input{15Simltion4}

\input{16Conclusions}
\section*{Acknowledgment}
We would like to express special thanks to Tyrus Berry from the Department of Mathematics at George Mason University, 
for listening and guiding us in regret upper bound analysis of POLA algorithm.
This work was partially supported by the NSF under grant
No. CNS-1502584 and CNS-1464487.
%We would like to express special thanks to Tyrus Berry for his patience listening and guiding our ideas and analysis on POLA algorithm and its regret upper bound analysis. 

\appendix{}
\input{17Appendix}

\def\bibfont{\footnotesize}
\bibliographystyle{unsrt}
\bibliography{Ref}

% that's all folks
\end{document}

%% file: 2Introduction.tex
\section{Introduction}

Nowadays the demand to wireless bandwidth is growing rapidly due to the increasing growth in various mobile and IoT applications, which raises the spectrum shortage problem.
To address this problem, Federal Communications Commission (FCC) has authorized opening spectrum bands (e.g., 3550-3700 MHz and TV white space) owned by licensed primary users (PU) to unlicensed secondary users (SU) when the primary users are inactive \cite{FCC2015, 4481339}.
Cognitive radio (CR) is a key technology that enables secondary users to learn the spectrum environment and dynamically access the best available channel.
Meanwhile, an attacker can send signals emulating the primary users to manipulate the spectrum environment, preventing a secondary user from utilizing the available channel.
This attack is called primary user emulation (PUE) attack \cite{gao2012security, 6195839, 5425707, 5585629, 5648777,HLZH2010, MDAA2016}.
%Fig.~\ref{SysModlfig} shows an example of a cognitive radio network with the presence of a PUE attacker.

Existing works on PUE attacks mainly focus on PUE attack detection \cite{Bian:2008:SVI:1554126.1554138, 4509846, 4413138} and defending strategies \cite{5425707, 5585629, 5648777, HLZH2010}. 
However, there is a lack of study on the optimal PUE attacking strategies.
Better understanding of the optimal attacking strategies will enable us to quantify the severeness or impact of a PUE attacker on the secondary user's throughput.
It will also shed light on the design of defending strategies.

In practice, an attacker may not have any prior knowledge of the primary user activity characteristics or the secondary user dynamic spectrum access strategies.
Therefore, it needs to learn the environment and attack at the same time.
In this paper, for the first time, we study the optimal PUE attacking strategies without any assumption on the prior knowledge of the primary user activity or secondary user accessing strategies.
We formulate this problem as an online learning problem. 
We also need to mention that the words \emph{play}, \emph{action taking}, and \emph{attack} are used interchangeably throughout the paper.
%%Similar idea holds between actions and channels.

Different from all the existing works on online learning based PUE attack defending strategies \cite{5425707, 5585629, 5648777, HLZH2010}, in our problem, an attacker cannot observe the reward on the attacked channel.
Considering a time-slotted system, the PUE attack usually happens in the channel sensing period, in which a secondary user attempting to access a channel conducts spectrum sensing to decide the presence of a primary user.
If a secondary user senses the attacked channel, it will believe the primary user is active so that it will not transmit the data in order to avoid interfering with the primary user.
In this case, the PUE attack is effective since it disrupts the secondary user's communication and affects its knowledge of the spectrum availability.
In the other case, if there is no secondary user attempting to access the attacking channel, the attacker makes no impact on the secondary user, so the attack is ineffective.
However, the attacker cannot differentiate between the two cases when it launches a PUE attack on a channel because no secondary user will be observed on the attacked channel no matter a secondary user has ever attempted to access the channel or not.

%For a PUE attacker to learn which channel a secondary user is most likely to access, it has to make observations.
The key for the attacker to launch effective PUE attacks is to learn 
which channel or channels a secondary user is most likely to access.
To do so, the attacker needs to make observations on the channels.
As a result, the attacker's performance is dependent on its observation capability.
We define the observation capability of the attacker as the number of the channels it can observe within the same time slot after launching the attack.
We propose two attacking schemes based on the attacker's observation capability.
%The idea is that since the reward on the attacked channel is unobservable, 
%during the rest of the time slot after attack period, 
%the PUE attacker may be capable of observing one or several other channels 
%than the attacked one to realize if there is any potential reward there.
In the first scheme called \emph{Attack-OR-Observe} (AORO), an attacker if attacks, cannot make any observation in a given time slot due to the 
short slot duration in highly dynamic systems \cite{MCKP2016} or when the channel switching time is long. 
We call this attacker an attacker with no observation capability within the same time slot (since no observation is possible if it attacks). 
To learn the most rewarding channel though, the attacker needs to dedicate some time slots for observation only without launching 
any attacks. 
On the other hand, an attacker could be able to attack a channel in the sensing phase and observe other channels 
in the data transmission phase within the same time slot if it can switch between channels fast enough. 
We call this attacking scheme, \emph{Attack-But-Observe-Another} (ABOA).
%, we propose two new attacking schemes called \emph{attack-or-observe} (AORO) and \emph{attack-but-observe-another} (ABOA), respectively.

%AORO method as mentioned is proposed for the cases when the attacker has no observation capability.
%Each time-slot cannot be divided to an attack phase and an observation phase.
%In AORO, the attacker has to decide to either attack or observe to identify potentially rewarding channels (the channels with most active secondary users) in each time slot.
%As the name of this method, Attack-OR-Observe, suggests only either attack or observation is feasible at each time slot.
%In ABOA, an attacker attacks a channel during the channel sensing phase but switches to at least one other channel to observe the secondary user's activity during the data transmission phase in the same slot.
%As explained earlier, this is the case when the attacker has at least one observation capability.
%This is why this scheme is called Attack-But-Observe-Another.
%This strategy is suitable when the time slot duration is large enough which gives the motivation that a short channel sensing phase is usually followed by a longer data transmission phase in which an attacker is able to switch to at least one other channel to observe the secondary user's activity. 
%TV white space spectrum sharing is an example of this category \cite{MCKP2016}.

For the AORO case, we propose an online learning algorithm called POLA -\emph{Play or Observe Learning Algorithm}- to dynamically decide between attack and observation 
and then to choose a channel for the decided action, in a given time slot.
%For attack purposes it chooses a channel dynamically based on the past observations and for observation purposes it chooses a channel uniform randomly.
Through theoretical analysis, we prove that POLA achieves a regret in the order of $\tilde{O}(\sqrt[3]{T^2})$ where
$T$ is the number of slots the CR network operates.
Its higher slope regret is due to the fact that it cannot attack (gain reward) and observe (learn) simultaneously in a given time slot.
We show the optimality of our algorithm by matching its regret upper bound with its lower bound.

For the ABOA case, we propose PROLA -\emph{Play and Random Observe Learning Algorithm}- as another online learning algorithm to be applied by the PUE attacker.  
PROLA dynamically chooses the attacking and observing channels in each time slot.
%It decides the observing channel uniform randomly.
%It also dynamically chooses the attacking channel in each slot according to the observed past activity of the secondary user.
The PROLA learning algorithm's observation policy fits a specific group of graphs called time-varying partially observable feedback graphs \cite{NANC2015}. 
It is derived in \cite{NANC2015} that these feedback graphs lead to a regret in the order of $\tilde{O}(\sqrt[3]{T^2})$. 
However, our algorithm, PROLA, is based on a new theoretical framework and we prove its regret is in the order of $\tilde{O}(\sqrt{T})$.
We then prove its regret lower bound is in the order of $\tilde{\Omega}(\sqrt{T})$ which matches its upper bound and shows our algorithms' optimality. 

Both algorithms proposed address the attacker's observation capabilities and 
can be applied as optimal PUE attacking strategies without any prior knowledge of the primary user activity and secondary user access
strategies. 

We further generalize PROLA to multi-channel observations where an attacker can observe multiple channels within the same time slot. 
By analyzing its regret, we come to the conclusion that 
%from one side, 
increasing the observation capability from one to multiple 
channels does not give additional benefit to the attacker in terms of regret order. 
It only improves the constant factor of the regret. 
%However, it gives insight on the approximate number of observations required to achieve a small constant factor.
	
Our main contributions are summarized as follows:
\begin{itemize}

\item We formulate the PUE attack as an online learning problem without any assumption on the prior knowledge of either primary user activity characteristics or secondary user dynamic channel access strategies.

\item We propose two attacking schemes, AORO and ABOA, that model the behavior of a PUE attacker. 
For the AORO case, a PUE attacker, in a given time slot, dynamically decides either to attack or observe; then chooses a channel for the decision it made.
While in the ABOA case, the PUE attacker dynamically chooses one channel to attack and chooses at least one other channel to observe within the same time slot. 

\item We propose an online learning algorithm POLA for the AORO case. 
POLA achieves an optimal learning regret in the order of $ \tilde{\Theta}(\sqrt[3]{T^2})$. 
We show its optimality by matching its regret lower and upper bounds.
% The algorithm and its analysis are further generalized to multichannel observation case.

\item For the ABOA case, we propose an online learning algorithm, PROLA, to dynamically decide the attacking and observing channels. 
%EXP3-DO has a suboptimal regret of $\tilde{O}(\sqrt[3]{T^2})$ and 
We prove that PROLA achieves an optimal regret order of $ \tilde{\Theta}(\sqrt{T})$ by deriving its regret upper bound and lower bound.
%This shows that with no observation at all in the attacking slot, the attacker loses on the regret order; However, with the observation of at least 
This shows that with an observation capability of at least one within the attacking slot, there is a significant improvement on the performance of the attacker.

\item The algorithm and the analysis of the PROLA are further generalized to multi-channel observations.
%Through our analysis we show that observation of multiple channels does not give additional benefit to the attacker (only a constant scaling) 
%though it gives insight on the number of observations required to achieve the minimum constant factor. 

\item Theoretical contribution: 
%We proposed a novel online learning algorithm, POLA, based on AORO scheme. 
%We proved this algorithm's optimality by deriving its regret upper and lower bounds and showing that they are in the same order.
%We show POLA is categorized as an optimal online learning algorithm by showing its regret upper bound matches its regret lower bound.
For PROLA, despite observing the actions partially, it achieves an optimal regret order of $ \tilde{\Theta}(\sqrt{T})$ 
which is better than a known bound of $ \tilde{\Theta}(\sqrt[3]{T^2})$. 
We accomplish it by proposing randomized time-variable feedback graphs.
%making random observations rather than deterministic ones. 

\item We conduct simulations to evaluate the performance of the proposed algorithms under various system parameters.

\end{itemize}

Through theoretical analysis and simulations under various system parameters, we have the following findings:

\begin{itemize}

\item With no observation at all within the attacking slot, the attacker loses on the regret order.
While with the observation of at least one channel, there is a significant improvement on the attacking performance.

\item Observation of multiple channels in the PROLA algorithm does not give additional benefit to the attacker in terms of regret order.
The regret is proportional to $\sqrt{1/m}$, where $m$ is the number of channels observed by the attacker.
Based on this relation, the regret is a monotonically decreasing and convex function of the number of observing channels.
As more observing channels are added, the reduction in the regret becomes marginal. 
Therefore, a relatively small number of observing channels is sufficient to approach a small constant factor.

\item The attacker's regret is proportional to $\sqrt{K\ln K}$, where $K$ is the total number of channels.
\end{itemize} 

%the number of observing channels has an important impact on the attacker's regret.

%Even though, all three learning algorithms proposed in this work can be generalized to $m$ observations, we show it only for PROLA algorithm.  
%The generalization of the two other algorithms is similar.

The rest of the paper is organized as follows.
Section \ref{RltdWrk} discusses the related work.
System model and problem formulation are described in Section \ref{Sys Mod Prblm Frmltion}.
Section \ref{New App} proposes the online learning-based access strategies applied by the attacker and derives the regret bounds.
Simulation results are presented in section \ref{Eval}.
Finally, section \ref{Conclusions} concludes the paper and discusses future work.
%The proof of the regret lower bounds are presented in the Appendix.

The proofs of the lower bounds have been moved to the Appendix as they closely follow existing lower bound proofs.
%We move some of the mathematical formulas and derivations 
%to the appendix to improve the readability and flow of the article, 
%and keep those that are indeed essential in order to give the 
%reader a fair chance to understand the results.
				
%\emph{In a nutshell, in this paper, we propose two novel online learning algorithms for the PUE attacker. 
%The first one, POLA, is a new algorithm addressing the learning problem when either observation or attack is feasible at each time. 
%We prove a regret order of $\tilde{O}(\sqrt[3]{T^2})$ for this algorithm.
%The second algorithm which is another novel online learning policy solves the problem when only side observations are possible. 
%Based on this algorithm, we advance the theoretical study on feedback graphs by proposing a strategy, PROLA, that achieves
 %the optimal regret in the order of $\tilde{O}(\sqrt{T})$ without observing the rewards on all the actions other than the playing one, simultaneously.
%In PROLA, the agent uniform randomly selects at least one action other than the playing one to observe at each time}.
%Then, our next algorithm, POLA, proposes a policy for the agent if it cannot do both at the same time step. 
%The agent gets a suboptimal regret in the order of \tilde{O}($\sqrt[3]{T^2}$)}.

%If the attacker acts based on the \cite{}, its regret will be in the order of $\tilde{O}(\sqrt[3]{T^2})$ eventhough it is making play and observe simultaneously at each time. 
%However, by applying our proposed algorithm PROLA it can achieve an optimal regret order.

%% file: 3RelatedWork.tex
\section{Related Work}
\label{RltdWrk}

\subsection{PUE Attacks in Cognitive Radio Networks}

Existing work on PUE attacks mainly focus on PUE attack detection \cite{Bian:2008:SVI:1554126.1554138, 4509846, 4413138} and defending strategies \cite{5425707, 5585629, 5648777, HLZH2010}.

There are few works discussing attacking strategies under  dynamic spectrum access scenarios.
In \cite{5425707, 5585629}, the attacker applies partially observable Markov decision process (POMDP) framework to find the attacking channel in each time slot.
It is assumed that the attacker can observe the reward on the attacking channel.
That is, the attacker knows if a secondary user is ever trying to access a channel or not.
In \cite{5648777, HLZH2010}, it is assumed that the attacker is always aware of the best channel to attack.
However, there is no methodology proposed or proved on how the best attacking channel can be decided.

The optimal PUE attack strategy without any prior knowledge of the primary user activity and secondary user access strategies is not well understood.
In this paper, we fill this gap by formulating this problem as an online learning problem.
Our problem is also unique in that the attacker cannot observe the reward on the attacking channel due to the nature of PUE attack.
\subsection{Multi-armed Bandit Problems}
There is a rich literature about online learning algorithms.
The most related ones to our work are multi-armed bandit (MAB) problems \cite{Auer:2002:FAM:599614.599677,PACB1995,PACB2003,KWQL2012,YGKB2014}. The MAB problems have many applications in cognitive radio networks with learning capabilities \cite{5648777, MDSS2016,QWML2015, MDAATCCN2017}.
In such problems, an agent plays a machine repeatedly and obtains a reward when it takes a certain action at each time.
Any time when choosing an action the agent faces a dilemma of whether to take the best rewarding action known so far or to try other actions to find even better ones.
Trying to learn and optimize his actions, the agent needs to trade off between exploration and exploitation.
On one hand the agent needs to explore all the actions often enough to learn which is the most rewarding one and on the other hand he needs to exploit the believed best rewarding action to minimize his overall regret.

%MAB problems can be categorized into stochastic ones and non-stochastic ones.
%In stochastic ones \cite{Auer:2002:FAM:599614.599677}, the reward of the actions is assumed under a parametrized distribution, but the parameter of the distribution is unknown.
%In non-stochastic ones \cite{PACB2003,PACB1995}, no probabilistic model is assumed for the reward and the reward can be arbitrary. Alternatively, the non-stochastic MAB problems can be interpreted as focused on some unknown sample path of reward.
%In our case, the PUE attacker does not have any prior knowledge of the primary user activity or secondary user channel accessing strategies, and consequently the non-stochastic MAB problem serves as a better framework.
For most existing MAB frameworks as explained, the agent needs to observe the reward on the taken action.
Therefore, these frameworks cannot be directly applied to our problem where a PUE attacker cannot observe the reward on the attacking channel.
Most recently, Alon et al. generalize the cases of MAB problems into feedback graphs \cite{NANC2015}.
In this framework, they study the online learning with side observations. They show in their work that if an agent takes an action without observing its reward, but observes the reward of all the other actions, it can achieve an optimal regret in the order of $\tilde{O}(\sqrt{T})$.
However, the agent can only achieve a regret of $\tilde{O}(\sqrt[3]{T^2})$ if it cannot observe the rewards on all the other actions simultaneously.
In other words, even one missing edge in the feedback graph leads to the regret of $\tilde{O}(\sqrt[3]{T^2})$.
%and if no action is left unobserved.
%based on the feedback graph

\emph{In this paper, we propose two novel online learning policies. 
The first one, POLA, for the case when only either observation or attack is possible within a time slot, is shown to achieve a regret in the order of $\tilde{O}(\sqrt[3]{T^2})$.
We then advance this theoretical study by proposing a strategy, PROLA, which is suitable for an attacker (learning agent) with higher observation capabilities within the acting time step.
We prove that PROLA achieves an optimal regret in the order of $\tilde{O}(\sqrt{T})$ without observing the rewards on all the channels other than the attacking one simultaneously.
In PROLA, the attacker uniformly randomly selects at least one channel other than the attacking one to observe in each time slot. 
Our framework is called randomized time-variable feedback graph}.
%Then we propose an optimal algorithm for the learning of an agent who can either play to gain reward or observe to learn the best action and prove its optimality}.
\subsection{Jamming Attacks}
There are several works formulating jamming attacks and anti-jamming strategies as online learning problems \cite{QWML2015, 5962525,QWKR2011, QWPX2011}.
In jamming attacks, an attacker can usually observe the reward on the attacking channel where an ongoing communication between legitimate users can be detected.
Also it is possible for the defenders to learn whether they are defending against a jammer or a PUE attacker by observing the reward on the accessed channel.
PUE attacks are different in that the attacker attacks the channel sensing phase and prevents a secondary user from utilizing an available channel.
As a result, a PUE attacker cannot observe the instantaneous reward on the attacked channel.
That is, it cannot decide if an attack is effective or not.

\subsection{Adaptive Opponents}
Wang et al. study the outcome/performance of two adaptive opponents when each of the agents applies a no regret learning-based access strategy \cite{QWML2015} .
One agent tries to catch the other on a channel and the other one tries to evade it.
Both learning algorithms applied by the opponents are no regret algorithms and are designed for oblivious environments. 
In other words, each of the learning algorithms is designed for benign environments and non of them assumes an adaptive/learning-based opponent.
It is shown in this work, that both opponents applying a no-regret learning-based algorithm, reach an equilibrium 
which is in fact a Nash Equilibrium in that case.
In other words, despite the fact that the learning-based algorithms applied are originally proposed for oblivious environments,
the two agents applying these algorithms act reasonably well to achieve an equilibrium.
Motivated by this work, in our simulations in Section \ref{Eval}, we have considered a learning-based secondary user.
More specifically, even though we proposed learning-based algorithms for oblivious environments, in the simulations, 
we evaluate its performance in an adaptive environment.
The simulation results show the rationality of the attacker that even against an adaptive opponent (learning-based cognitive radio), it performs well and achieves a regret below the derived upper bound.

%% file: 4Problemformulation.tex
\section{System Model and Problem Formulation}
\label{Sys Mod Prblm Frmltion}

We consider a cognitive radio network consisting of several primary users, multiple secondary users, and one attacker.
There are $K$ ($K>1$) channels in the network.
We assume the system is operated in a time-slotted fashion.
	          
\subsection{Primary User}
In each time slot, each primary user is either active (on) or inactive (off).
We assume the on-off sequence of PUs on the channels is unknown to the attacker a priori. In other words, the PU activity can follow any distribution or can even be arbitrary.
\subsection{Secondary User}
The secondary users may apply any dynamic spectrum access policy \cite{KESS2013, PLTL2010, YLTW2015}.
In each time slot, each SU conducts spectrum sensing, data transmission, and learning in three consecutive phases as shown in Fig.~\ref{timeslotfig}(a).

At the beginning of each time slot, each secondary user senses a channel it attempts to access.
If it finds the channel idle (i.e., the primary user is inactive), then accesses this channel; Otherwise, it remains silent till the end of the current slot in order to avoid interference to the primary user.
At the end of the slot, it applies a learning algorithm to decide which channel it will attempt to access in the next time slot based on its past channel access experiences.

We assume the secondary users cannot differentiate the attacking signal from the genuine primary user signal.
That is, in a time slot, when the attacker launches a PUE attack on the channel a secondary user attempts to 
access, the secondary user will not transmit any data on that channel.
\begin{figure}
						  \centering
								\includegraphics[scale=0.35]{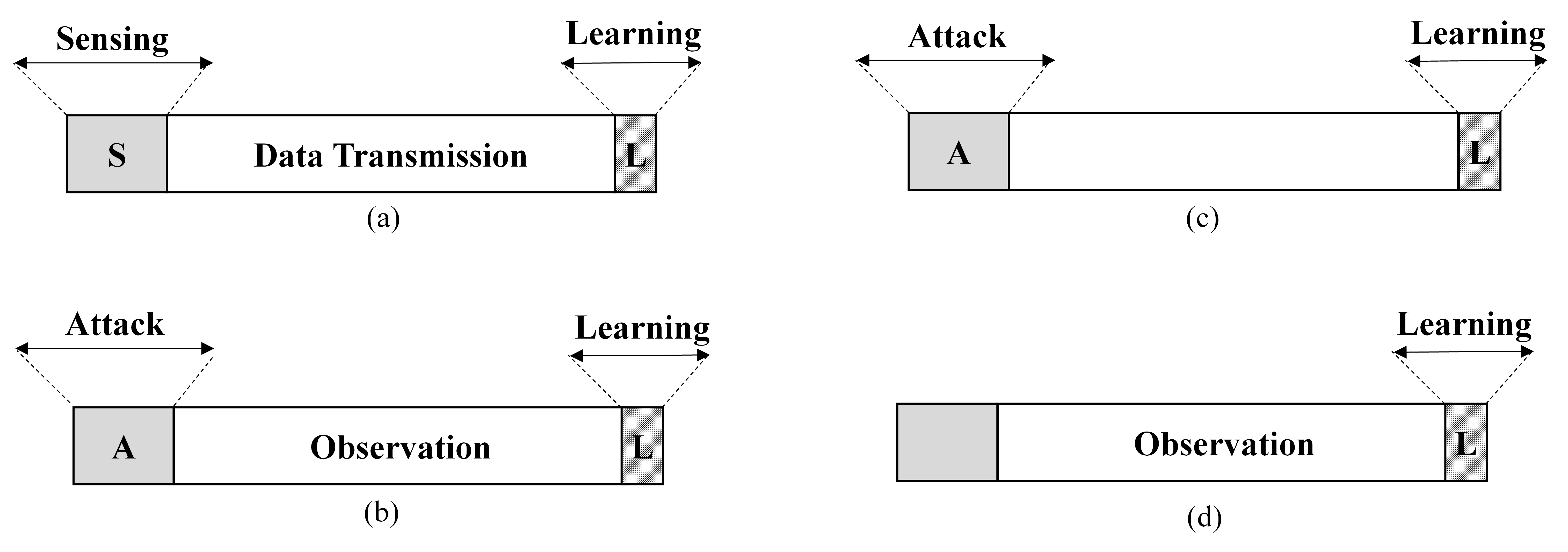}
							\caption{Time slot structure of a) an SU and b,c,d) a PUE attacker.}
							\label{timeslotfig}
						\end{figure}
\subsection{Attacker}
We assume a smart PUE attacker with learning capability.
We do not consider attack on the PUs.
The attacker may have different observation capabilities.
Observation capability is defined as the number of channels the attacker can observe after the attack period within the same time-slot.
Observation capabilities of the attacker are impacted by the overall number of it's antennas, time slot duration, switching time, etc.

Within a time-slot, an attacker with at least one observation capability conducts three different operations.
First in the attacking phase, the attacker launches the PUE attack by sending signals emulating the primary user's signals \cite{5585629} to attack the SUs.
Note that, in this phase, the attacker has no idea if its attack is effective or not.
That is, it does not know if a secondary user is ever trying to access the attacking channel or not.
In the observation phase, the attacker is supposed to observe the communication on at least one other channel.
The attacker can even observe the attacked channel in the observation phase, however, in that case it will sense nothing since it has attacked the channel by emulating the PU signal and scared away any potential SU attempting to access the channel. So from learning point of view, observing the attacked channel does not provide any useful information on the SU’s activity. 
Therefore, it only gets useful information when observing channels other than the attacked one.
Observing a different channel, the attacker may detect a primary user signal, a secondary user signal, or nothing on the observing channel. 
The attacker applies its observation in the learning phase to optimize its future attacks.
At the end of the learning period, it decides which channels it attempts to attack and observe in the next slot.	
Fig.~\ref{timeslotfig}(b) shows the time slot structure for an attacker with at least one observation capability.

If the attacker has no observation capability within the attacking time slot, at the end of each time slot it still applies a learning strategy based on which, 
it decides whether to dedicate the next time slot for attack or observation, then chooses a channel for either of them.
Fig.~\ref{timeslotfig}(c,d) show the time slot structure for an attacker with no observation capability in the attacking time slot.
As shown in these two figures, the attacker conducts either attack or observation at each time slot.

\subsection{Problem Formulation}
\label{ProbForm}
%Given the above system model, the challenge for the attacker is to determine when and where to attack and observe. \emph{which channel to launch attack and which channel to observe}.
Since the attacker needs to learn and attack at the same time and it has no prior knowledge of the primary user activity or secondary user access strategies, we formulate this problem as an online learning problem.

We consider $T$ as the total number of time slots the network operates.
We define $x_t(j)$ as the attacker's reward on channel $j$ at time slot $t$ ($ 1 \leq j \leq K $, $1\leq t\leq T$).
Without loss of generality, we normalize $x_t(j)\in\left[0,1\right]$. More specifically:
\begin{align}
\label{rwrdDef}
x_t(j) =
\begin{cases}
1, & \text{\emph{SU is on channel $j$ at time $t$}} \\
0, & o.w.
\end{cases}.
\end{align}

\begin{table}
\centering
\caption{Main Notation}
\begin{tabular}{|c|l|} \hline
$T$& total number of time slots  \\ %\hline
$K$        & total number of channels    \\ %\hline
$k$        & index of the channel   \\ %\hline
$I_t$      & index of the channel to be attacked at time $t$   \\ %\hline
$J_t$      & index of the channel to be observed at time $t$  \\ %\hline
$R $       & total regret of the attacker in Algorithm 1  \\ %\hline
$\gamma$   &exploration rate used in algorithm 2         \\ %\hline
$\eta$     &learning rate used in Algorithms 1 and 2       \\ %\hline
$p_t\left(i\right)$   & attack distribution on channels at time $t$ in Algorithms 1 and 2     \\ %\hline
$q_t\left(i\right)$   & observation distribution on channels at time $t$ in Algorithm 1   \\ %\hline
$\omega_t\left(i\right)$   & weight assigned to channel $i$ at time $t$ in Algorithms 1 and 2    \\ %\hline
$\delta_t$   & observation probability at time $t$ in Algorithm 1     \\ \hline
\end{tabular}
\end{table}

Suppose the attacker applies a learning policy $\varphi$ to select the attacking and observing channels.
The aggregated expected reward of attacks by time slot $T$ is equal to
\begin{equation}
	\label{equ6}
				{G_\varphi(T)} = \bold{E}_\varphi\left[\sum\limits_{t=1}^T x_t(I_t)\right], \\
\end{equation}		
where $I_t$ indicates the channel chosen at time $t$ to be attacked. 
The attacker's goal is to maximize the expected value of the aggregated attacking reward, thus to minimize the throughput of the secondary user,
\begin{equation}
\label{equ7}
\mathrm{maximize} \quad G_{\varphi}(T).
\end{equation}

For a learning algorithm, regret is commonly used to measure its performance.
The regret of the attacker can be defined as follows
\begin{equation}
\label{equ8}
Regret =  G_{max} - G_{\varphi}\left(T\right),
\end{equation}
where
\begin{equation}
G_{max} = \max\limits_j \sum\limits_{t=1}^T x_{t}\left(j\right).
\end{equation}

The regret measures the gap between the accumulated reward achieved applying a learning algorithm and the maximum accumulated reward the attacker can obtain when it keeps attacking the single best channel.
Single best channel is the channel with highest accumulated reward up to time $T$ \cite{PACB2003}.
%This regret is usually called the weak regret \cite{PACB2003}.
Then the problem can be transformed to minimize the regret
\begin{equation}
\label{equ10}
\mathrm{minimize} \quad G_{max} - G_{\varphi}\left(T\right).
\end{equation}
Table I summarizes the main notation used in this paper.

%% file: 5OptAttckStrtgy.tex
\section{Online Learning-Based Attacking Strategy}
\label{New App}

In this section we propose two novel online learning algorithms for the attacker.
These algorithms do not require the attacker to observe the reward on the attacked channel.
Moreover, our algorithms can be applied in any other application with the same requirements.
%of not observing the reward on action that is taken but being able to observe the reward on one other action whether it is on the same time frame or in the next time frame.

The first algorithm proposed, POLA, is suitable for an attacker with no observation capability in the attacking time slot .
For this attacker, either attack or observation is feasible within each time slot.
Based on this learning strategy, the attacker decides at each time slot whether to attack or observe and chooses a channel for either of them, dynamically.
The assumption here is that, any time the attacker decides to make observation, it observes only one channel since time slot duration has been considered short or switching costs are high compared to the time slot duration.
%If it decides to make an observation, the attacker chooses a channel uniformly at random. only one channel since attack and observation at the same time slot is infeasible for such an attacker.

The second algorithm, PROLA, is proposed for an attacker with at least one observation capability. 
Based on this learning policy, at each time, the attacker chooses channels dynamically for both attack and observation.
%Based on the attacker's observation capabilities, we assume the attacker can observe the reward on at least one other channel.
In the following, we assume the attacker's observation capability is one. % can observe the reward on only one other channel.
We generalize it to multiple channel observation capability in Section \ref{MultiactionObsOPT}.

Both algorithms are considered as no-regret online learning algorithms in the sense that the incremental regret between two consecutive time slots diminishes to zero as time goes to infinity.
The first one, POLA, achieves an optimal regret in the order of $\tilde{O} (\sqrt[3]{T^2})$. 
This higher slope arises from the fact that at each time slot, 
%the attacker either avoids gaining reward to spend the time slot on observation to learn to make better decisions in the future, or it gains some reward without having made enough observations which makes its choices 
the attacker gains only some reward by attacking a channel without updating its learning or it learns by making observation without being able to gain any reward. 
%EXP3-DO is also a suboptimal learning algorithm with regret in the order of $\tilde{O} (\sqrt[3]{T^2})$ and the latter,
PROLA, is also optimal with the regret proved in the order of $\tilde{O}(\sqrt{T})$.
Comparing these algorithms and the generalization of the latter together shows that, changing from no observation capability in the same time slot to one observation capability, there is a significant improvement in the regret order; 
this is in comparison to changing from one observation capability to multiple observation capability which does not give any benefits in terms of regret order.
However, multiple channel observation capability, provides insight to 
find the appropriate number of observations required to achieve the minimum constant factor in the regret upper bound.

%In the next sections, we provide the proofs for those theorems and corollaries that are more relevant to the topic of the current work within the body of the paper and postpone those requiring more abstract math into the Appendix.
 
\begin{figure}
						  \centering
								\includegraphics[scale=.48]{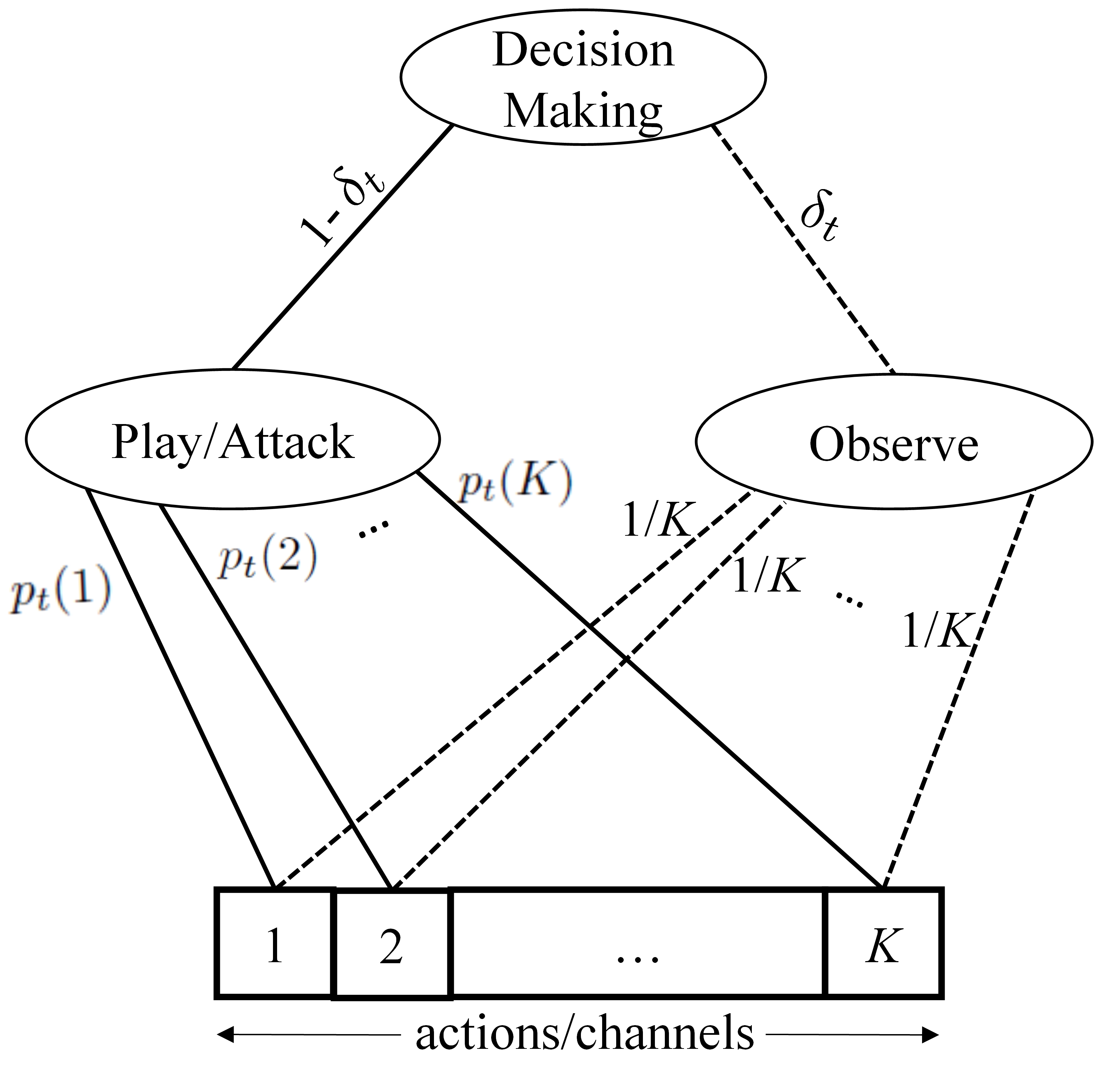}
							\caption{Decision Making process based on the POLA strategy: First the agent decides between play and observe, then for either of them chooses a channel dynamically.}
							\label{POLAfig}
	\end{figure}

%% file: 6OptmlLearningPOLA.tex
\subsection{Attacking Strategy 1: POLA Learning Algorithm}
\label{ProposedPOLA}

POLA is an online learning algorithm for the learning of an agent with no observation capability within the attacking time slot. 
In other words, the agent cannot play and observe simultaneously.
If modeled properly as a feedback graph as is shown in Appendix \ref{APPNDX}, this problem can be solved based on the EXP3.G algorithm presented in \cite{NANC2015}. 
However, our proposed learning algorithm, POLA, leads to a smaller regret constant compared to the one in \cite{NANC2015} and it is much easier to understand.
Based on the POLA attacking strategy, the attacker at each time slot decides between attacking and observation, dynamically.
If it decides to attack, it applies an exponential weighting distribution based on the previous observations it has made.
It chooses a channel uniformly at random if it decides to make an observation.
Figure~\ref{POLAfig} represents the decision making structure of POLA.
The proposed algorithm is presented in Algorithm 1.

Since the attacker is not able to attack and observe within the same time slot simultaneously, it needs to trade off between attacking and observation cycles. 
From one side the attacker needs to make observations to learn the most rewarding channel to attack. 
From the other side, it needs to attack often enough to minimize his regret.
So, $\delta_t$ which represents the trade off between attack and observation needs to be chosen carefully.

In step 2 of Algorithm 1, $\hat{x}_t(j)$ represents an unbiased estimate of the reward $x_t(j)$. In order to derive this value, we divide the $x_t(j)$ by the probability that this channel is chosen to be observed which is equal to $\delta_t /K$.
This is because in order for each channel to be chosen to be observed, first the algorithm needs to decide to observe with probability $\delta_t$
and then to choose that channel for observation with probability $\frac{1}{K}$.

For any agent that acts based on the Algorithm 1, the following theorem holds. 

\begin{thm}
For any $K \geq 2$ and for any $\eta \leq \sqrt[3]{\frac{ \ln K }{K^2 T}}$, the upper bound on the expected regret of Algorithm 1 

\begin{align}
G_{max} - \bold{E}\left[G_{POLA}\right]  &\leq (e-2)\eta K \left( \frac{3}{4} \sqrt[3]{\frac{\left(T+1\right)^4}{K \ln K}} + \frac{K \ln K}{4}\right) \nn
&+ \frac{\ln K}{\eta} + \frac{3}{2} \sqrt[3]{K T^2 \ln K } 
\end{align}

holds for any assignment of rewards for any $T>0$.
\end{thm}

\noindent\rule{8.82cm}{0.5pt}
\textbf{Algorithm 1: POLA, Play or Observe Learning Algorithm }\\
\noindent\rule{8.82cm}{0.5pt}

\hangindent=6.5 em
\hangafter=1
\textbf{Parameter:} $\eta \in \left(0, \sqrt[3]{\frac{ \ln K }{K^2 T}} \right]$.  
%g \geq T

\textbf{Initialization:} $\omega_1(i) =1$, \quad  $i = 1,..., K $. 

\textbf{For each} $t = 1,2,...,T $ 

1. Set $ \delta_t = \min \large \left\{ 1 , \sqrt[3]{ \frac{ K \ln K }{t} } \right\} $. \\ \\
With probability $ \delta_t $ decide to observe and go to step 2.\\
\hangindent=2.4 em
\hangafter=1 
With probability $ 1-\delta_t $ decide to attack and go to step 3.

2. Set 
\[ q_t(i) = \frac{1}{K} , \quad i = 1,..., K \] 
\hangindent=2.4 em
Choose $J_t \sim\ q_t$ and observe the reward $x_t(J_t) $. \\ 
For $j = 1, 2,..., K$
\[
\hat{x}_t(j) =
 \begin{cases}
 \frac{  x_t(j)}{\delta_t \left(1/K\right)}, & j = J_t \\
0, & o.w.,
\end{cases}
\]
\[\omega_{t+1}( j ) = \omega_t(j) \exp (\eta \hat{x}_t(j)),\]
\hangindent=2.4 em
Go back to step 1. 

3. Set 
\[p_t(i)= \frac{\omega_t(i)}{\sum \limits_{j=1}^K \omega_t(j)}  , \quad i = 1,..., K \]  
\hangindent=2.4 em
Attack channel $I_t \sim\ p_t$ and accumulate the unobservable reward $x_t(I_t) $.\\ 
For $j = 1, 2,..., K$
\[ \hat{x}_t(j) = 0,\]
\[ \omega_{t+1}( j ) = \omega_t(j), \]
\hangindent=2.4 em
Go back to step 1.

\hangindent = 0 em
\noindent\rule{8.82cm}{0.5pt}

%% file: 7RegretUbPOLA.tex
\begin{proof}[Proof of Theorem 1]

The regret at time $t$ is a random variable equal to, 

\[
r(t) =
\begin{cases}
x_t(j^*) - \bold{E}_ {POLA,A} \left[x_t(I_t)\right]  & ,\text{Attack} \\
x_t(j^*) - \bold{E}_ {POLA,O} \left[x_t(I_t)\right]  = x_t(j^*) & ,\text{Observe}  
\end{cases}
\]

where $j^\ast$ is the index of the best channel.
% and the expectation is w.r.t. the randomness in reward generation on the channels.
The expected value of regret at time $t$ is equal to
\begin{align}
\label{equ200}
\bold{E}\left[r(t)\right] &=  \left(1- \delta_t \right) \left(x_t(j^*) - \bold{E}_ {POLA,A} \left[x_t(I_t)\right]\right)  + \delta_t \: x_t(j^*),   \nn
\end{align}
where the expectation is w.r.t. to the randomness in the attacker's attack policy.
Expected value of accumulated regret, $R$, is
\begin{align}
\label{equ201}
\bold{E} \left[R\right] &= \bold{E} \left[ \sum\limits_{t=1}^T r(t)  \right] \nn
%&= \sum\limits_{t=1}^T E\left[  r(t)  \right] \nn
&= \sum\limits_{t=1}^T  \left(1- \delta_t \right) \left(x_t(j^*) - \bold{E}_ {POLA,A} \left[x_t(I_t)\right]\right)  \nn
&\quad +\sum\limits_{t=1}^T \delta_t \: x_t(j^*)   \nn
%&= \sum\limits_{t=1}^T  \left(1- \delta_t \right) x_t(j^*) - \sum\limits_{t=1}^T \left(1- \delta_t \right) \bold{E}_ {POLA,A} \left[x_t(I_t)\right]  \nn
%& +\sum\limits_{t=1}^T \delta_t \: x_t(j^*)   \nn
&\leq \sum\limits_{t=1}^T   x_t(j^*) - \sum\limits_{t=1}^T  \bold{E}_ {POLA,A} \left[x_t(I_t)\right] + \sum\limits_{t=1}^T \delta_t. \nn
%& = G_{max} - \bold{E} [G_{POLA,A}] + \sum\limits_{t=1}^T \delta_t
\end{align}

The inequality results from the fact that $\delta_t \geq 0 $ and $ x_t(j^*) = 1 $. % which means $1- \delta_t \leq 1 $. 
%So we substitute $1- \delta_t $ with its upper bound. 
It is also assumed that $x_t(i) \leq 1$ for all $i$.
Now the regret consists of two terms. 
The first one is due to not attacking the most rewarding channel but attacking some other low rewarding channels.
The second term arises because of the observations in which the attacker gains no reward. 
We derive an upper bound for each term separately, then add them together. 

$\delta_t$ plays the key role in minimizing the regret.
First of all, $\delta_t$ needs to be decaying since otherwise it leads to a linear growth of regret.
So the key idea in designing a no-regret algorithm is to choose an appropriate decaying function for $\delta_t$.

From one side, slowly decaying $\delta_t $ is desired from learning point of view; 
however, it results in a larger value of regret since staying more in the observation phase precludes the attacker from launching attacks and gaining rewards.
On the other hand, if $\delta_t$ decays too fast, the attacker settles in a wrong channel since it does not have enough time to learn the most rewarding channel.
We choose $ \delta_t = \min \large \left\{ 1 , \sqrt[3]{ \frac{ K \ln K }{t} } \right\} $ and our analysis shows the optimality of this function in minimizing the regret.

%Following similar steps as in \cite{PANC2003,NANC2015} and considering that the upper bound 
%on $\sum\limits_{t=1}^T \frac{1}{\delta_t}$ is equal to $\frac{3}{4} \sqrt[3]{\frac{(T+1)^4}{K \ln K}} + \frac{K \ln K }{4}$, 
%the upper bound on the first term in Eq. (\ref{equ201}) can be derived as follows.
Below is the derivation of the upper bound for the first term of regret in equation (\ref{equ201}). 
For a single $t$
\begin{align}
\label{equ202}
\frac{W_{t+1}}{W_t} 
%&= \sum\limits_{i=1}^{K}\frac{\omega_{t+1}(i)}{W_t} \nn
&=\sum\limits_{i=1}^{K}\frac{\omega_t(i)}{W_t} \exp (\eta \hat{x}_{t}(i)) \nn
&= \sum\limits_{i=1}^{K} p_t(i) \exp (\eta \hat{x}_{t}(i)) \nn
&\leq \sum\limits_{i=1}^{K} p_t(i) [1+\eta \hat{x}_t(i)+(e-2) \eta^2 \hat{x}_t^2(i)] \nn
%&= 1 + \eta \sum\limits_{i=1}^K p_t(i)\hat{x}_t(i) + (e-2) \eta^2 \sum\limits_{i=1}^K p_t(i)\hat{x}_t^2  (i)\nn
&\leq \exp \bigg( \eta \sum\limits_{i=1}^K p_t(i)\hat{x}_t(i) + (e-2) \eta^2 \sum\limits_{i=1}^K p_t(i)\hat{x}_t^2 (i) \bigg),
\end{align}
where the equalities follow from the definition of $W_{t+1} = \sum_{i=1}^K \omega_{t}\left(i\right)$, $\omega_{t+1}\left(i\right)$, and $p_t(i)$, respectively in Algorithm 1. 
Also the last inequality follows from the fact that $e^x \geq 1+x$. 
Finally, the first inequality holds since $e^x \leq 1+x+(e-2)x^2$ for $x \leq 1$. 
In this case, we need $\eta \hat{x}_t(i) \leq 1$. 
Based on our algorithm, $\eta \hat{x}_t(i) = 0$ if $i \neq J_t$ and for $i = J_t$ we have 
$\eta \hat{x}_t(i) = \eta \frac{x_t(i)}{ \delta_t \frac{1}{K} } \leq \eta K \sqrt[3]{\frac{T} {K \ln K}} $, 
since $x_t(i) \leq 1$ and $\delta_t \geq \sqrt[3]{\frac{K \ln K}{T}}$ for $T \geq K \ln K$. 
This is equivalent to $\eta \leq \frac{1}{K \sqrt[3]{\frac{T}{K \ln K }}} = \sqrt[3]{\frac{ \ln K }{K^2 T}}$.

By taking the $\ln $ and summing over $t=1$ to $T$ on both sides of equation (\ref{equ202}), the left hand side (LHS) of the equation will be equal to 
%\begin{equation}
%\label{equ203}
%\ln \frac{W_{t+1}}{W_t} \leq 
%\eta \sum\limits_{i=1}^K p_t(i) \hat{x}_t(i) + (e-2) \eta^2 \sum\limits_{i=1}^K p_t(i) \hat{x}_i^2(t)
%\end{equation}
%Then sum over $t=1:T$. For the left hand side (LHS) of equation (\ref{equ203}) we have, 
\begin{align}
\label{equ204}
\sum\limits_{t=1}^T \ln \frac{W_{t+1}}{W_t} &= \ln \frac{W_{T+1}}{W_1} \nn
&= \ln W_{T+1} - \ln K \nn
%&= \ln \sum\limits_{j=1}^K \omega_{T+1}(j) - \ln K  \nn
&\geq \ln \omega_{T+1}(j) - \ln K \nn
%&= \ln \exp (\eta \sum\limits_{t=1}^T \hat{x}_t(j)) - \ln K   \nn
&= \eta \sum\limits_{t=1}^T \hat{x}_t(j) - \ln K.
\end{align}

By combining (\ref{equ204}) with (\ref{equ202}), 

\begin{align}
\label{equ205}
\eta & \sum\limits_{t=1}^T \hat{x}_t(j) - \ln K \nn
& \leq \eta \sum\limits_{t=1}^T \sum\limits_{i=1}^K p_t(i) \hat{x}_t(i) + (e-2)\eta^2 \sum\limits_{t=1}^T \sum\limits_{i=1}^K  p_t(i) \hat{x}_t^2(i).
\end{align}

We take the expectation w.r.t. the randomness in $\hat{x}$, substitute $j$ by $j^\ast$ since $j$ can be any of the actions, use the definition of $E_{POLA,A}[x_t(I_t)]$, and with a little simplification and rearranging the equation we get the following,
%\begin{align}
%\label{equ206}
%\eta \sum\limits_{t=1}^T x_t(j) - \ln K &\leq
%\eta \sum\limits_{t=1}^T \sum\limits_{i=1}^K p_t(i) x_t(i) \nn 
%&+ (e-2)\eta^2 K \sum\limits_{t=1}^T \sum\limits_{i=1}^K  p_t(i) \frac{x_t^2(i)}{\delta_t}
%\end{align}
%\begin{align}
%\label{equ207}
%\sum\limits_{t=1}^T x_t(j) - \frac{\ln K}{\eta} &\leq
%\sum\limits_{t=1}^T \sum\limits_{i=1}^K p_t(i) x_t(i) \nn
%&+ (e-2)\eta K \sum\limits_{t=1}^T \sum\limits_{i=1}^K  p_t(i) \frac{x_t^2(i)}{\delta_t}
%\end{align}
%\begin{align}
%\label{equ208}
%\sum\limits_{t=1}^T  x_t(j^*) - \frac{\ln K}{\eta} &\leq \sum\limits_{t=1}^T  E_{POLA} \left[x_t(I_t)\right] \nn
% &+ (e-2)\eta K \sum\limits_{t=1}^T \sum\limits_{i=1}^K  p_t(i) \frac{1}{\delta_t}
%\end{align}
%\begin{align}
%\label{equ209}
%\sum\limits_{t=1}^T x_t(j^*) - \sum\limits_{t=1}^T  E_{POLA} \left[x_t(I_t)\right]
% &\leq (e-2)\eta K \sum\limits_{t=1}^T \frac{1}{\delta_t} \sum\limits_{i=1}^K  p_t(i) \nn
% &+ \frac{\ln K}{\eta}
%\end{align}
%Then, 

\begin{align}
\label{equ210}
\sum\limits_{t=1}^T x_t(j^*) - \sum\limits_{t=1}^T & E_{POLA,A} \left[x_t(I_t)\right]  \nn
&\leq  (e-2)\eta K \sum\limits_{t=1}^T \frac{1}{\delta_t} + \frac{\ln K}{\eta}.
\end{align}

The upper bound on $\sum\limits_{t=1}^T \frac{1}{\delta_t}$ is equal to $\frac{3}{4} \sqrt[3]{\frac{(T+1)^4}{K \ln K}} + \frac{K \ln K }{4}$  which gives us, 
\begin{align}
\label{equ211}
\sum\limits_{t=1}^T & x_t(j^*) - \sum\limits_{t=1}^T \bold{E}_{POLA,A}  \left[x_t(I_t)\right]  \nn
& \leq \frac{\ln K}{\eta} + (e-2)\eta K  \left( \frac{3}{4} \sqrt[3]{\frac{\left(T+1\right)^4}{K \ln K}} + \frac{K \ln K}{4}\right). 
\end{align}

The upper bound on the second term of equation (\ref{equ201}) is
\begin{align}
\label{equ212}
\sum\limits_{t=1}^T \delta_t 
%& 
= \sum\limits_{t=1}^T \min \large\left\{1, \sqrt[3]{\frac{K \ln K}{t}}\right\} 
%\nn
%&= \sum\limits_{t=1}^{K \ln K} 1 +  \sum\limits_{t={K \ln K +1 }}^T \sqrt[3]{\frac{K \ln K}{t}} \nn
%&= K \ln K + \sqrt[3]{K \ln K} \sum\limits_{t={K \ln K +1 }}^T \frac{1}{\sqrt[3]{t}} \nn
%& \leq K \ln K + \sqrt[3]{K \ln K} \int\limits_{K \ln K }^T \frac{1}{\sqrt[3]{t}} dt \nn
%& \leq K \ln K + \frac{3}{2} \sqrt[3]{K \ln K} \left[T^{(2/3)}- (K \ln K )^{(2/3)}\right] \nn
%&= \frac{3}{2} \sqrt[3]{K T^2 \ln K } - \frac{1}{2} K \ln K  \nn
%&
 \leq \frac{3}{2} \sqrt[3]{K T^2 \ln K }.
\end{align}

Summing up the equation (\ref{equ211}) and equation (\ref{equ212}) gives us the regret upper bound.
\end{proof}

By choosing appropriate values for $\eta$, the above upper bound on the regret can be minimized. 

\begin{cor}
\label{cor1}
 For any $T > 2.577 K \ln K $, we consider the input parameter  \\
\[
 \eta = \sqrt{\frac{\ln K}{ (e-2) K \left( \frac{3}{4} \sqrt[3]{\frac{\left(T+1\right)^4}{K \ln K}} + \frac{K \ln K}{4}   \right)}    }.
\]
Then
\[
%\begin{align}
G_{max} - \bold{E} [ G_{POLA} ] \leq (\sqrt{3 (e-2) }  + \frac{3}{2} ) \sqrt[3]{  T ^ 2 K \ln K } 
%\nn & \leq  2.97 \sqrt[3]{ T ^ 2 K \ln K }
\]
%\end{align}
holds for any arbitrary assignment of rewards.
\end{cor}

\begin{proof}[Proof of Corollary 1]

We sketch the proof as follows.
By getting the derivative, we find the optimal value for $\eta$.
Since $\eta \leq \sqrt[3]{\frac{ \ln K }{K^2 T}}$, in order for the regret bound to hold, we need 
$T \geq \frac{8}{3 \sqrt{3 (e-2)^3}}  K \ln K = 2.577 K \ln K $. 
By plugging in the value of $\eta$ and some simplifications the regret bound in the corollary is achieved. 

%will be equal to,
%\begin{align}
%&G_{max} - \bold{E} [ G_{POLA} ] \nn
%& \leq \frac{3}{2}  \sqrt[3]{ T ^ 2 K \ln K } + \sqrt{e-2} \sqrt{K \ln K} \sqrt{  3 \sqrt[3]{\frac{(T+1)^4}{K \ln K }  }  + K \ln K  }.
%\end{align}
%If we consider $ 3 \sqrt[3]{\frac{(T+1)^4}{K \ln K }  }  >> K \ln K $ which is equivalent to $T >> \frac{K \ln K }{\sqrt[4]{27}} - 1 = 0.4 K \ln K - 1$ then the regret bound in the corollary is achieved.
\end{proof}

%To apply the Corollary, $g$ can be replaced by $T$ since all the rewards are in $[0,1]$ and the attacker attacks for a duration of $T$ time steps.

Next is a theorem on the regret lower bound for this problem under which attacking and observation are not possible simultaneously.

\begin{thm}
For $K \geq 2$ and for any player strategy $A$, the expected weak regret of algorithm $A$ is lower bounded by

\[
G_{max} - \bold{E}\left[G_A\right]  \geq  v \sqrt[3]{K T^2} \\
\]
for some small constant $v$ and it holds for some assignment of rewards for any $T>0$.
\end{thm}

\begin{proof}[Proof of Theorem 2]
This follows from the lower bound in \cite{NANC2015}. The construction is given in the Appendix..
\end{proof}

Based on Theorems 1 and 2, the algorithm's regret upper bound matches its regret lower bound which indicates POLA is an optimal online learning algorithm.

%% file: 8OptmlLearningPROLA.tex
\subsection{Attacking Strategy 2: PROLA Learning Algorithm}

In this section, we propose another novel online learning algorithm that 
can be applied by the PUE attacker with one observation capability within the attacking slot.
The proposed optimal online learning algorithm, called PROLA, 
at each time, chooses two actions dynamically to play and observe, respectively.
The action to play is chosen based on an exponential weighting, 
while the other action for observation is chosen uniformly at random excluding the played action. 
The proposed algorithm is presented in Algorithm 2.

Figure \ref{FdbckGrphfig} shows the observation policy governing the actions for $K=4$ in a feedback graph format. 
In this figure, $Y_{ij}(t)$ is an observation indicator of channel $j$ when channel $i$ is attacked at time $t$.
We define $Y_{ij}(t) \in \left\{0,1\right\}$ such that at each time for the chosen action $i$ to be played, 
\begin{align}
\label{Equ600}
\sum \limits_{j=1, j\neq i} ^ K Y_{ij}(t) = 1,   \quad   i = 1 , \ldots, K \quad t = 1, \ldots,T.   %1 \leq m \leq K-1$.
\end{align}

\noindent\rule{8.82cm}{0.5pt}
\textbf{Algorithm 2 : PROLA, Play and Random Observe Learning Algorithm}\\
\noindent\rule{8.82cm}{0.5pt}

\hangindent=6.5 em
\hangafter=1
\textbf{Parameters:}  $\gamma \in \left( 0, 1\right)$, \\  
$\eta \in  \left( 0,   \frac{\gamma}{2(K-1)}                 \right]$. \\

\textbf{Initialization:} $\omega_1(i) = 1$, \quad $i = 1,...,K.$

\textbf{For each} $t = 1,2,...,T$ \\

1. Set
$p_t(i)= (1-\gamma ) \frac{\omega_t(i)}{\sum_{j=1}^K\omega_t(j)} + \frac{\gamma}{K}$, \quad\quad $i=1,...,K$.  \\

\hangindent=2.4em
\hangafter=1
2. Attack channel $I_t \sim\ p_t$ and accumulate the unobservable reward $x_t(I_t) $. \\

\hangindent=2.4em
\hangafter=1
3. Choose a channel $J_t$ other than the attacked one uniformly at random and observe its reward $x_t(J_t)$ based on equation (\ref{rwrdDef}).\\

4. For $j = 1,...,K$ 
\[
{\hat{x}_t}(j) =
\begin{cases}
\frac{x_t(j)}{(1/(K-1))(1 - p_t(j))}, & j = J_t \\
0, & o.w.,
\end{cases}
\]
\indent \quad\quad\quad\quad $\omega_{t+1}(j)=\omega_t(j) \exp (\eta \hat{x}_t(j))$.\\
\noindent\rule{8.82cm}{0.5pt} 

In other words, there is a policy that for any channel $i$ played, only one of the observation indicators takes a value of one and the rest take a value of zero. 
For example, if channel 2 is attacked ($i=2$), only one of the three outgoing edges from 2 will be equal to one.
This edge selection policy represents the channel selection process for observation.
We define it as a uniform random distribution equal to $\frac{1}{K-1}$. 
We call this feedback graph, a time-variable random feedback graph.
%The values on the edges indicate the observation probabilities.
Our feedback graph fits into the time-variable feedback graphs introduced in \cite{NANC2015} and based on the results derived in that work, 
the regret upper bound of our algorithm is $\tilde{O}( \sqrt[3]{T^2})$.
However, based on our analysis, the upper bound on the attacker's regret is in the order of $\tilde{O}( \sqrt{T})$
 which shows a significant improvement. 
In other words, despite the fact that the agent makes only partial observation on the channels, 
it achieves a significantly improved regret order compared to the no observation in the attacking slot case.
This has been possible due to the new property we considered in the partially observable graphs which is adding randomness. 
In the long run, randomness makes full observation possible to the agent.

\begin{figure}
\centering
\includegraphics[width=.4\textwidth]{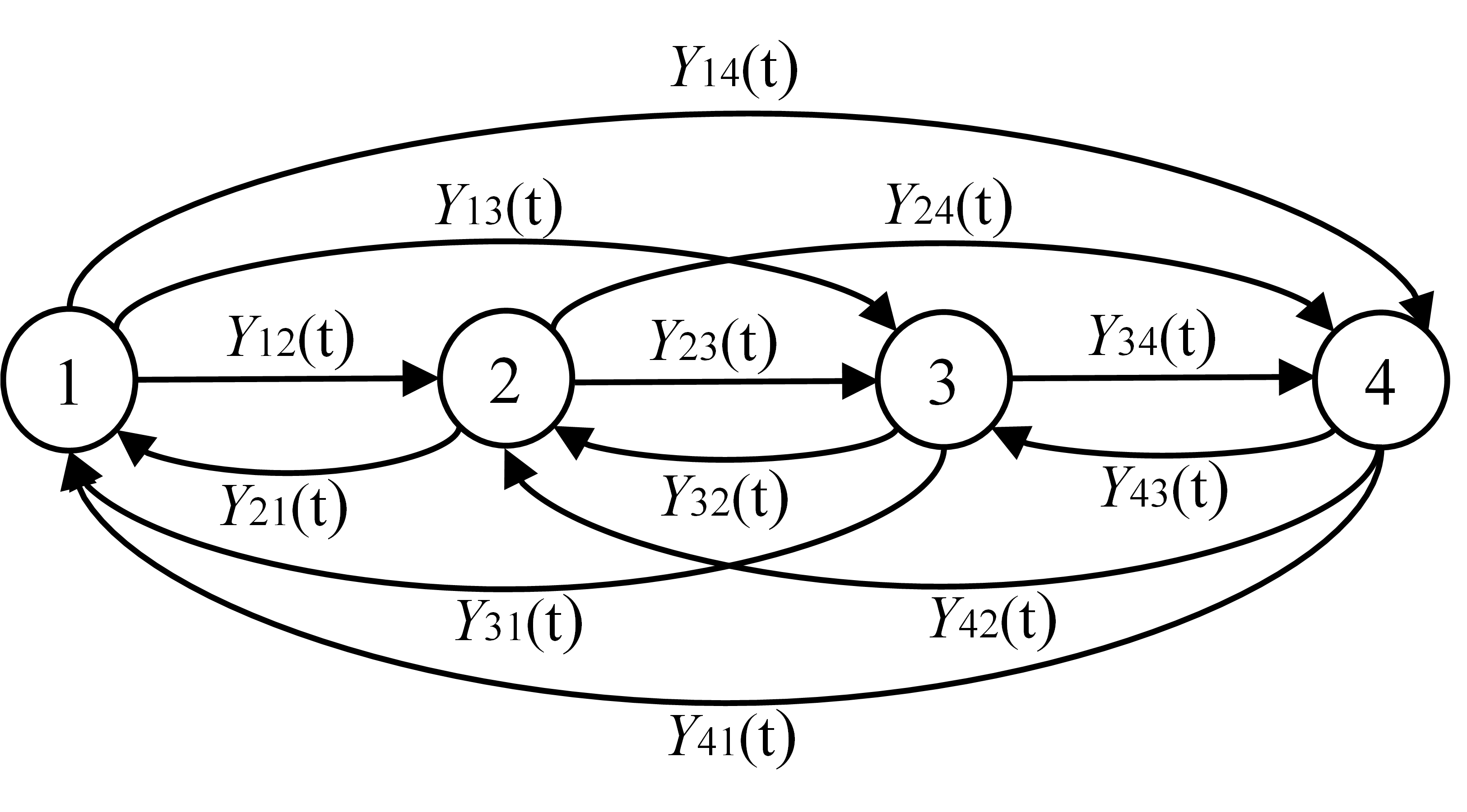}
\caption{Channel observation strategy, $K= 4$}
\label{FdbckGrphfig}
\end{figure}

In Step 4 of Algorithm 2, in order to create $\hat{x}_t(j)$, an unbiased estimate of the actual reward $x_t(j)$, we divide the observed reward, $x_t(J_t)$, by $(1/(K-1))(1 - p_t(J_t))$which is the probability of choosing channel $J_t$ to be observed. In other words, channel $J_t$ will be chosen to be observed if it has not been chosen for attacking ,($1 - p_t(J_t)$), and second if it gets chosen uniformly at random from the rest of the channels, ($1/(K-1)$).

\begin{thm}
For any $K \geq 2$ and for any $\eta \leq \frac{\gamma}{2(K-1)}$, for the given randomized observation structure for the attacker the upper bound on the expected regret of Algorithm 2, 
\[
G_{max}-\bold{E}[G_{PROLA}]\leq  (e-2) (K-1)  \frac{\eta}{1- \gamma }  G_{max} +  \frac{\ln K}{\eta (1- \gamma)}
\]
holds for any assignment of rewards for any $T>0$.
\end{thm}

%% file: 9RegretUbPROLA.tex
\begin{proof}[Proof of Theorem 3]

By using the same definition for $W_t = \omega_t(1) + \cdots +\omega_t(K) = \sum\limits_{i=1}^K \omega_t(i) $ as in Algorithm 1, at each time $t$,

\begin{align}
\label{equ11}
\frac{W_{t+1}}{W_t} 
%&= \sum\limits_{i=1}^{K}\frac{\omega_{t+1}(i)}{W_t} \nn
%&=\sum\limits_{i=1}^{K}\frac{\omega_t(i)}{W_t} \exp (\eta \hat{x}_{t}(i)) \nn
&= \sum\limits_{i=1}^{K}\frac{p_t(i)-\gamma/K}{1-\gamma} \exp (\eta \hat{x}_{t}(i)) \nn
&\leq \sum\limits_{i=1}^{K}\frac{p_t(i)-\gamma/K}{1-\gamma} [1+\eta \hat{x}_t(i)+(e-2)(\eta \hat{x}_t(i))^2] \nn
%&\leq 1 + \frac{\eta}{1-\gamma}\sum\limits_{i=1}^K p_t(i)\hat{x}_t(i) \nn
%&\quad\quad+\frac{(e-2) \eta^2}{1-\gamma} \sum\limits_{i=1}^K p_t(i)\hat{x}_t^2  (i)\nn
&\leq \exp \bigg(\frac{\eta}{1-\gamma}\sum\limits_{i=1}^K p_t(i)\hat{x}_t(i) \nn
&\quad\quad\quad\quad+\frac{(e-2) \eta^2}{1-\gamma} \sum\limits_{i=1}^K p_t(i)\hat{x}_t^2 (i) \bigg).
\end{align}

The equality follows from the definition of $W_{t+1}$, $\omega_{t+1}\left(i\right)$, and $p_t(i)$ respectively in Algorithm 2. 
Also, the last inequality follows from the fact that $e^x \geq 1+x$. Finally, the first inequality holds since $e^x \leq 1+x+(e-2)x^2$ for $x \leq 1$. 
When $\eta \leq \frac{\gamma}{2(K-1)}$, the result, $\eta \hat{x}_t(i) \leq 1$, follows from the observation that either $\eta \hat{x}_t(i) = 0$ or $\eta \hat{x}_t(i) = \eta \frac{x_t(i)}{\frac{1}{K-1} (1-p_t(i))} \leq \eta (K-1) \frac{2}{\gamma} \leq 1$, since $x_t(i) \leq 1$ and $p_t(i)= (1-\gamma ) \frac{\omega_t(i)}{\sum_{j=1}^K\omega_t(j)} + \frac{\gamma}{K} \leq 1-\gamma + \frac{\gamma}{2} \leq 1- \frac{\gamma}{2}$.

%Now we take the logarithm of both sides of (\ref{equ11}) and we have %and use the fact that $1+x \leq e^x$\\
%\begin{align}
%\label{equ13}
%\ln \frac{W_{t+1}}{W_t} \leq &
%\frac{\eta}{1-\gamma} \sum\limits_{i=1}^K p_t(i) \hat{x}_t(i) \nn
%&+ \frac{(e-2) \eta^2}{1-\gamma} \sum\limits_{i=1}^K p_t(i) \hat{x}_i^2(t).
%\end{align}

By taking the logarithm of both sides of (\ref{equ11}) and summing over $t$ from $1$ to $T$, we derive the following inequality on the LHS of the equation,
\begin{align}
\label{equ14}
\sum\limits_{t=1}^T \ln \frac{W_{t+1}}{W_t} &= \ln \frac{W_{T+1}}{W_1} \nn
%&= \ln W_{T+1} - \ln K \nn
%&= \ln \sum\limits_{j=1}^K \omega_{T+1}(j) - \ln K  \nn
&\geq \ln \omega_{T+1}(j) - \ln K \nn
%&= \ln \exp (\eta \sum\limits_{t=1}^T \hat{x}_t(j)) - \ln K   \nn
&= \eta \sum\limits_{t=1}^T \hat{x}_t(j) - \ln K.
\end{align}
%By combining (\ref{equ14}) with (\ref{equ13}) we have:
%\begin{align}
%\label{equ15}
%\eta \sum\limits_{t=1}^T \hat{x}_t(j) - \ln K \leq
%&\frac{\eta}{1-\gamma} \sum\limits_{t=1}^T \sum\limits_{i=1}^K p_t(i) \hat{x}_t(i) \nn
%&+ \frac{(e-2)\eta^2}{1-\gamma} \sum\limits_{t=1}^T \sum\limits_{i=1}^K  p_t(i) \hat{x}_t^2(i).
%\end{align}

Combining (\ref{equ11}) and (\ref{equ14}), we can get 
%By combining (\ref{equ14}) with (\ref{equ13}) and then reorganization it, we have
\begin{align}
\label{equ36}
\sum\limits_{t=1}^T & \hat{x}_t(j) -  \sum\limits_{t=1}^T \sum\limits_{i=1}^K p_t(i)\hat{x}_t(i) \nn
& \leq \gamma \sum\limits_{t=1}^T \hat{x}_t(j) + (e-2) \eta \sum\limits_{t=1}^T \sum\limits_{i=1}^K p_t(i)\hat{x}_t^2 (i) + \frac{\ln K}{\eta}.
\end{align}

Let $\acute{x}_t(i) = \hat{x}_t(i) - f_t$  where $f_t = \sum\limits_{i=1}^K p_t(i)\hat{x}_t(i)$. We make the pivotal observation that \eqref{equ36} also holds for $\acute{x}_t(i) $ since $\eta \acute{x}_t(i) \leq 1$, which is the only key to obtain \eqref{equ36}. 

%Therefore,
%\begin{align}
%\label{equ38}
%\sum\limits_{t=1}^T \acute{x}_t(j)
%-& \sum\limits_{t=1}^T \sum\limits_{i=1}^K  p_t(i) \acute{x}_t(i) \leq \gamma \sum\limits_{t=1}^T \acute{x}_t(j) + \nn
%&(e-2) \eta \sum\limits_{t=1}^T \sum\limits_{i=1}^K  p_t(i) \acute{x}_t^2(i) +\frac{\ln K}{\eta}.
%\end{align}
We also note that,
\begin{align}
\label{equ37}
 \sum\limits_{i=1}^K p_t(i) \acute{x}_t^2(i) & = \sum\limits_{i=1}^K p_t(i) (\hat{x}_t(i) - f_t)^2 \nn
%& = \sum\limits_{i=1}^K p_t(i) \hat{x}_t^2(i) + \sum_{i=1}^K p_t(i) f_t^2 - 2 \sum_{i=1}^K p_t(i) \hat{x}_t(i) f_t \nn
& = \sum\limits_{i=1}^K p_t(i) \hat{x}_t^2(i) - f_t^2 \nn
& \leq \sum\limits_{i=1}^K p_t(i) \hat{x}_t^2(i) - \sum\limits_{i=1}^K p_t^2(i) \hat{x}_t^2(i) \nn
& = \sum\limits_{i=1}^K p_t(i) (1-p_t(i)) \hat{x}_t^2(i).
\end{align}
%We can substitute, $\acute{x}_t(j) =  \hat{x}_t(j) - E[ \hat{x}_t(j) ]  $ in this equation which results in.

Substituting $\acute{x}_t(i)$ in equation (\ref{equ36}) and combining with \eqref{equ37}, 
\begin{align}
\label{equ39}
&\sum\limits_{t=1}^T (\hat{x}_t(j) - f_t) - \sum\limits_{t=1}^T \sum\limits_{i=1}^K  p_t(i) (\hat{x}_t(i) - f_t) \nn
& = \sum\limits_{t=1}^T \hat{x}_t(j)  - \sum\limits_{t=1}^T \sum\limits_{i=1}^K  p_t(i) \hat{x}_t(i) \nn
&\leq \gamma \sum\limits_{t=1}^T \left(\hat{x}_t(j)- \sum\limits_{i=1}^K p_t(i)\hat{x}_t(i)\right) \nn
& + (e-2) \eta \sum\limits_{t=1}^T \sum\limits_{i=1}^K  p_t(i) (1-p_t(i))  \hat{x}_t^2(i) + \frac{\ln K}{\eta}.
\end{align}

Observe that $\hat{x}_t(j)$ is similarly designed as an unbiased estimate of $x_t(j)$. Then for the expectation with respect to the sequence of channels attacked by the horizon $T$, 
\[
\bold{E}[\hat{x}_t(j)] = x_t(j), \bold{E}\left[\sum_{i=1}^K p_t(i) \hat{x}_t(i)\right] = \bold{E}[x_t(I_t)],
\]
and
\begin{align*}
& \bold{E} \left[\sum_{i=1}^K p_t(i)(1-p_t(i)) \hat{x}_t^2(i)\right] \\
& = \bold{E} \left[\sum_{i=1}^K p_t(i)(K-1)x_t(i) \hat{x}_t(i)\right] \leq K-1.
\end{align*}

We now take the expectation with respect to the sequence of channels attacked by the horizon $T$ in both sides of the last inequality of \eqref{equ39}. For the left hand side, 
\begin{align}
\label{equ40}
\bold{E} \left[\sum\limits_{t=1}^T \hat{x}_t(j)\right]  -& \bold{E} \left[\sum_{t=1}^T \sum\limits_{i=1}^K  p_t(i) \hat{x}_t(i)\right] \nn
%\sum\limits_{t=1}^T (x_t(j) - E[ x_t(j) ] )-& \sum\limits_{t=1}^T \sum\limits_{i=1}^K  p_t(i) (x_t(i) - E[x_t(i)] )= \nn
%&\sum\limits_{t=1}^T x_t(j)  - \sum\limits_{t=1}^T \sum\limits_{i=1}^K  p_t(i)
%(x_t(i) \nn
&= G_{max} - \bold{E} [G_{PROLA}].
\end{align}
and for the right hand side,
\begin{align*}
\label{equ41}
& \bold{E} [\mathrm{R.H.S}] \nn
&\leq \gamma (G_{max} - \bold{E}[G_{PROLA}]) + (e-2) (K-1) \eta  T + \frac{\ln K}{\eta}
\end{align*}

Combining the last two equations we get,
\[
(1-\gamma) (G_{max} - \bold{E} [G_{PROLA}]) \leq (e-2) (K-1) \eta  T + \frac{\ln K}{\eta},
\]
and since $G_{max}$ can be substituted by $T$, by rearranging the relation above the regret upper bound is achieved. 

% $\gamma \leq 1/2$.
%\[
%G_{max} - E[G_{PROLA}] \leq 2 (e-2) (K-1) \eta  T + \frac{\ln K}{\eta},
%\]
%assuming $\gamma \leq 1/2$. This concludes the proof.
\end{proof}

Similarly, we can minimize the regret bound by choosing appropriate values for $\eta$ and $\gamma$.

\begin{cor}
For any $T \geq \frac{8(K-1) \ln K}{e-2}$ and $\gamma = \frac{1}{2}$, we consider the following value for $\eta$ \\

\quad\quad\quad\quad\quad\quad$\eta = \sqrt{\frac{\ln K }{2 (e-2) (K-1) T }}$. \\
Then
\[
G_{max}-\bold{E}[G_{PROLA}]\leq 2 \sqrt{2(e - 2)} \sqrt{T (K-1) \ln K }
\]
holds for any arbitrary assignment of rewards.
\end{cor}

\begin{proof} [Proof of Corollary 2]
We sketch the proof as follows. 
By getting the derivative, we find the optimal value for $\eta$.
Since $\eta \leq \frac{\gamma}{2(K-1)}$, in order for the regret bound to hold, we need 
$T \geq \frac{8(K-1) \ln K}{e-2}$.  
Replacing the value of $\eta$ gives us the result on the regret upper bound.
%$g$ can be replaced by $T$ since all the rewards are in $[0,1]$ and the network runs for $T$ time slots, .
\end{proof}

\emph{ The important observation is that, based on \cite{NANC2015} such an algorithm, Algorithm 2, 
is expected to achieve a regret in the order of $\tilde{O}\left(\sqrt[3]{T^2}\right)$ since it can be categorized 
as a partially observable graph. 
However, our analysis gives a tighter bound and shows 
not only it is tighter but also it achieves the regret order of fully observable graphs.
This significant improvement has been accomplished by introducing randomization into feedback graphs.}

%This shows that randomization can be an important factor in designing online learning algorithms.
The following theorem provides the regret lower bound for this problem.
\begin{thm}
For any $K \geq 2$ and for any player strategy A, the expected weak regret of algorithm A is lower bounded by
\[
G_{max}-\bold{E}[G_{A}] \geq c \sqrt{KT} 
\]
for some small constant $c$ and it holds for some assignment of rewards for any $T>0$.
\end{thm}

\begin{proof}[Proof of Theorem 4]
The proof follows closely the lower bound proof in \cite{PANC2003}. Details are provided in the Appendix.
\end{proof}

%% file: 10MultiAcExtensionPROLA.tex
\subsection{Extension to Multiple Action Observation Capability}
\label{MultiactionObsOPT}

We generalize Algorithm PROLA to the case of an agent with multiple observation capability. 
This is suitable for an attacker with multiple observation capability when the attacker after the attack phase can observe multiple other channels within the same time-slot. 
%The agent can observe more than only one action at each time. 
%The number of actions that can be observed by the agent other than the played action depends on the . 
At least one and at most $K-1$ observations are possible by the agent (attacker).
%on the physical properties of the attacker, like the number of antennas he is equipped with, the required switching time between channels. It also 
In this case, $m$ indicates the number of possible observations and $ 1 \leq m \leq K-1$. 
Then, for $m$ observations, we modify equation (\ref{Equ600}) as follows,
\begin{align}
\label{Equ900}
\sum \limits_{j=1, j\neq i} ^ K Y_{ij}(t) = m,   \quad   i = 1 , \ldots, K \quad t = 1, \ldots,T .
\end{align}

The probability of a uniform choice of observations at each time, $\sum \limits_{j=1, j\neq i} ^ K Y_{ij}(t) = m$
is equal to $\frac{1}{{{K-1} \choose {m}}}$.
Corollary $3$ shows the result of the analysis for $m$ observations.

\begin{cor}
If the agent can observe $m$ actions at each time, then the regret upper bound is equal to 
\[
G_{max}-\bold{E}\left[G_{A}\right]\leq 4 \sqrt{\left(e - 2\right)} \sqrt{T \frac{K-1}{m} \ln K }.
\]
This regret upper bound shows a faster convergence rate when making more observations.
\end{cor}

\begin{proof} [Proof of Corollary 3]
We substitute $1/\left(K-1\right)$ by $m/\left(K-1\right)$ in Step 4 of Algorithm PROLA since at each time, $m$ actions are being chosen uniformly at random. Then the regret upper bound can be derived by similar analysis.
\end{proof}

As our analysis shows, making multiple observations does not improve the regret in terms of its order in $T$ compared to the case of making only one observation.
This means that only one observation and not more than that is sufficient to make a significant improvement in terms of regret order reduction compared to the case of no observation within the attacking slot.
The advantage of making more observations however is in reducing the constant coefficient in the regret. 
Making more observations leads to a smaller constant factor in the regret upper bound.
This relationship is non-linear though. 
i.e., the regret upper bound is proportional to $\sqrt{\frac{1}{m}} $ which means 
that in order to make a large reduction in the regret in terms of its constant coefficient impact, 
only a few observations would suffice. 
%makes a large reduction in the regret in terms of its coefficient impact
%will be sufficient to approach the smallest constant factor. 
%in order to get the smallest constant factor, one does not need
%to make full observation on all the channels rather making only a few observations will be sufficient to get the smallest constant factor. 
Our simulation results in Section \ref{mObservations} provide more details on this non-linear relationship.

%% file: 11Simulation.tex
\section{Performance Evaluation}
\label{Eval}

In this section, we present the simulation results to evaluate the validity of the proposed online learning algorithms applied by the PUE attacker. 
All the simulations are conducted in MATLAB and the results achieved are averaged over 10,000 independent random runs.

We evaluate the performance of the proposed learning algorithms, POLA and PROLA, 
and compare them with their theoretical regret upper bounds $\tilde{O}(\sqrt[3]{T^2})$ and $\tilde{O}(\sqrt{T})$, respectively.
POLA and PROLA correspond to an attacker with no observation and one observation capability within the attacking time slot, respectively. 
We then examine the impact of different system parameters on the attacker's performance.
The parameters include the number of time slots, total number of channels in the network, and the distribution on the PU activities.
We also examine the performance of an attacker with multiple observation capabilities. 
Finally we evaluate a secondary user's accumulated traffic with and without the presence of a PUE attacker.

$K$ primary users are considered, each acting on one channel. The primary users' on-off 
activity follows a Markov chain or i.i.d. distribution in the network.
Also, the PU activities on different channels are independent from each other.
$K$ idle probabilities are generated using MATLAB's rand function, each denoting one PU activity on each channel if PUs follow an i.i.d. Bernoulli distributions. 
$pI = [0.85	\ 0.85 \ 0.38 \	0.51 \	0.21 \	0.13 \	0.87 \	0.7 \	0.32 \ 0.95] $ is a vector of $K$ elements each of which denotes 
the corresponding PU activity on the $K$ channels.
If the channels follow Markov chains, for each channel, we generate three probabilities, $p01$, $p10$, and $p1$ as the transition probabilities from state 0 (on) to 1 (off), from 1 to 0, and the initial idle probability, respectively.
The three vectors below, are examples considered to represent PU activities. 
$p01 = [0.76	\ 0.06 \	0.3 \	0.24 \	0.1 \	0.1 \	0.01 \ 0.95 \	0.94 \	0.55] $, 
$p10 = [0.14 \	0.43 \	0.23 \	0.69 \	0.22 \	0.59 \	0.21 \	0.58 \	0.34 \	0.73] $, 
and $p1 = [0.53	\ 0.18 \	0.88 \	0.66 \	0.23 \	0.87 \	0.48 \	0.44 \	0.45 \	0.88]$. 
%The generated probabilities remain the same throughout the simulations. 

The PUE attacker employs either of the proposed attacking strategies, POLA, or PROLA. 
Throughout the simulations, when we talk about an attacker employing PROLA, we assume the attacker's observation capability is one within the attacking slot, unless otherwise stated.

Since the goal here is to evaluate the PUE attacker's performance, for simplicity we consider one SU in the network. 
Throughout the simulations, we assume the SU employs an online learning algorithm called Hedge \cite{PACB1995}. 
The assumption on the Hedge algorithm is that the secondary user is able to observe the rewards on all the channels in each time slot. %(PU activities) 
Hedge provides the minimum regret in terms of both order and the constant factor among all optimal learning-based algorithms. 
As a result the performance of our proposed learning algorithms can be evaluated in the worst case scenario for the attacker. 
As explained in Section \ref{RltdWrk}, even though in our analysis we considered an oblivious environment, in the simulations we can consider an SU that runs Hedge (an adaptive opponent). This experimental setup adds value since it shows that our proposed online learning-based algorithms 
perform reasonably even against adaptive opponents by keeping the occurred regret below the theoretical bounds derived. 

%In fact, this experimental setup adds value to the paper, as it shows that the algorithm performs reasonably also against some adaptive opponents.
\begin{figure*}[t]
\centering
\subfigure[POLA Learning]{
\includegraphics[width=.31\textwidth]{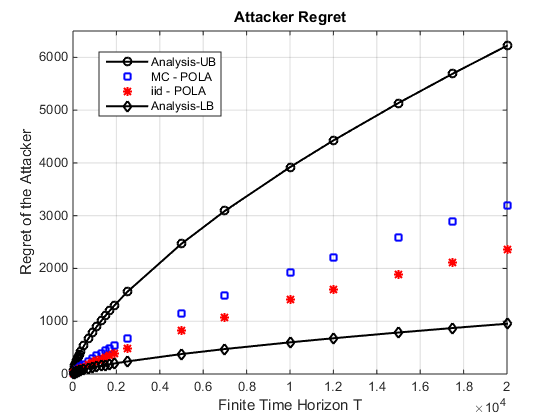}
}
\subfigure[POLA, channel impact, i.i.d. PU]{
\includegraphics[width=.31\textwidth]{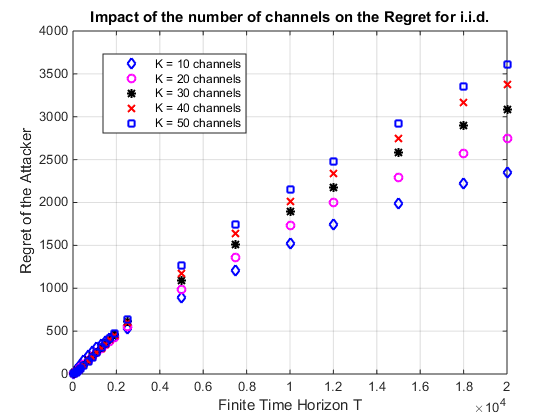}
}
\subfigure[POLA, channel impact, M.C. PU]{
\includegraphics[width=.31\textwidth]{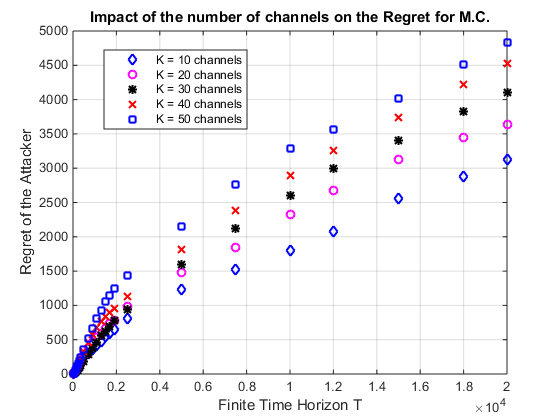}
}

\subfigure[PROLA Learning]{
\includegraphics[width=.31\textwidth]{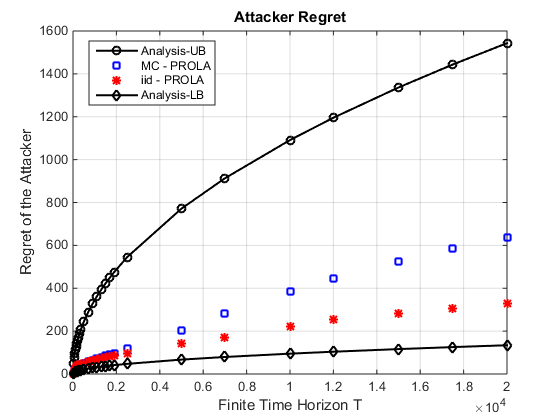}
}
\subfigure[PROLA, channel impact, i.i.d. PU]{
\includegraphics[width=.31\textwidth]{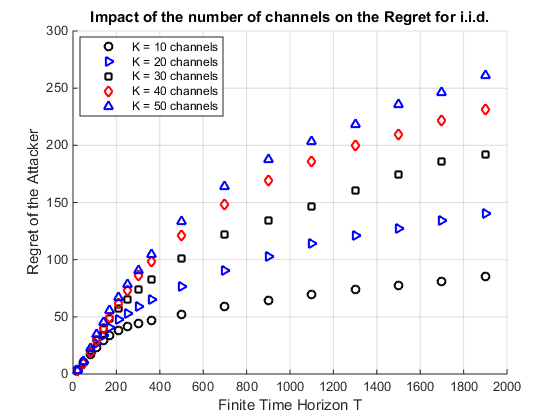}
}
\subfigure[PROLA, channel impact, M.C. PU]{
\includegraphics[width=.31\textwidth]{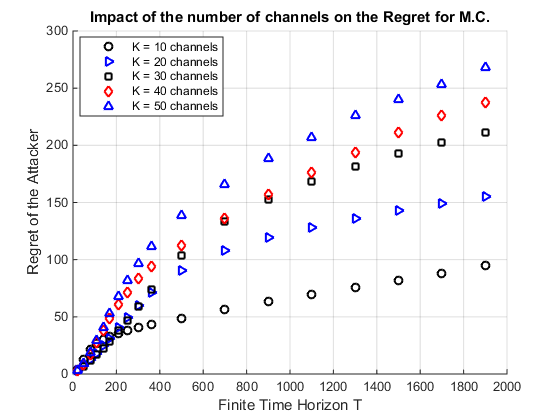}
}

\subfigure[PROLA, channel impact, i.i.d. PU]{
\includegraphics[width=.31\textwidth]{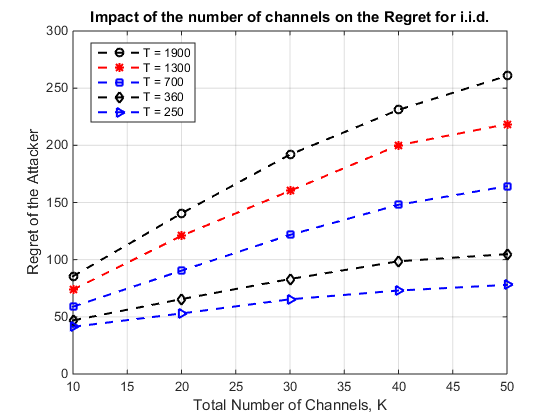}
}
\subfigure[PROLA, observation impact, i.i.d. PU]{
\includegraphics[width=.31\textwidth]{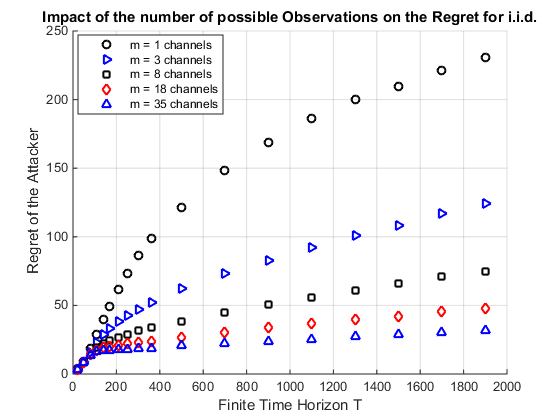}
}
\subfigure[PROLA, observation impact, M.C. PU]{
\includegraphics[width=.31\textwidth]{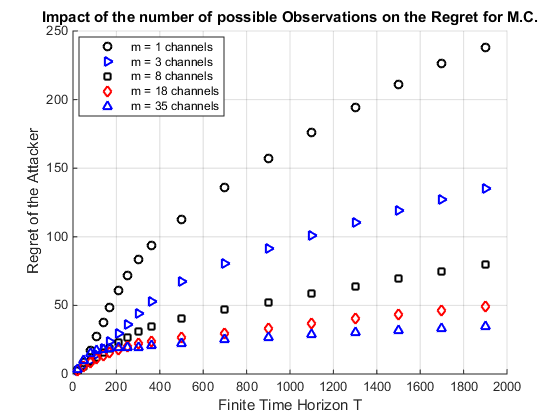}
}
\caption{Simulation results under different PU activity assumptions}
\label{Simluton}
\end{figure*}

%% file: 12Simltion1.tex
\subsection{Performance of POLA and the impact of the number of channels}	
\label{NumbrChnnlsPOLA}
In this section, we evaluate the performance of the proposed algorithm, POLA and compare it with the theoretical analysis.
Figure~\ref{Simluton}(a) shows the overall performance of the attacker for both cases of i.i.d. and Markovian chain distribution of PU activities on the channels for $K=10$. 
Regret upper bound and lower bound from the analysis in Corollary $1$ and Theorem $2$, respectively are also plotted in this figure. 
Then we examine the impact of the number of channels in the network on it's performance.
We consider $K$ variable from $10$ to $50$.
The attacker's regret when PUs follow i.i.d. distribution and Markovian chain for different number of actions are shown in Fig.~\ref{Simluton}(b) and (c), respectively. 
Since POLA has a regret with higher slope, in order to better observe the results, we have plotted the figures for $T$ from $1$ to $20,000$. 
We can observe the following from the figures.

\begin{itemize}

\item Regardless of the PUs activity type, the occurred regret is below the regret upper bound achieved from the theoretical analysis. 
%This can be seen by plugging in the values of $K$, and $T$ in the regret upper bounds derived in Theorems $1$.

\item The regret on all three figures has a higher slope compared to the results for PROLA in Fig.~\ref{Simluton}(d)-(f)
which also complies with our analysis. The higher slope in these figures can be seen by comparing their x and y axises.

\item As the number of channels increases the regret increases as is expected based on the analysis from Corollary 1. 
%The dependency on $K$ is 

\item As the number of channels increases, the regret does not increase linearly with it. 
Instead, the increment in the regret becomes marginal which complies with the theoretical analysis. 
Based on Corollary 1 the regret is proportional to $(\sqrt[3] {K \ln K})$. 
The dependency on $K$ can be represented by plotting the regret versus $K$ as is shown in the next subsection.
\end{itemize}

%\subfigure[Accumulated Traffic of an SU]{
%\includegraphics[width=.31\textwidth]{figures/WthANDout.png}

%% file: 13Simltion2.tex
\subsection{Performance of PROLA and the impact of the number of channels}	
\label{NumbrChnnls}

We compare the performance of the proposed learning algorithm, PROLA,  with the theoretical analysis from section \ref{New App}. 
We consider a network of $K = 10$ channels. 
Fig.~\ref{Simluton}(d) shows the simulation results as well as analytical results. 
From the simulations, we observe that the actual regret occurred in the simulations, 
is between the bounds achieved from the analysis regardless of the type of PU activity which complies with our analysis.
We also note that when we derived the theoretical upper bound we did not make any assumption on the PU activity. 
The regret is only dependent on the $K$ and $T$.

Then, we examine the impact of the number of channels in the network on the attacker's performance when it applies PROLA. 
Figure~\ref{Simluton}(e) and (f) show the attacker's regret when PUs follow i.i.d. distribution and Markovian chain respectively for $K$ variable from $10$ to $50$.
The same discussion on the system parameters and the results hold as in subsection \ref{NumbrChnnlsPOLA}.
In another representation, we plot the regret versus the number of channels, $K$, as the $x$ axis in Fig.~\ref{Simluton}(g).
%In this figure, $T=1900$ is considered.

Moreover, comparing regret values in Fig.~\ref{Simluton}(a) with those in Fig.~\ref{Simluton}(d), we observe the huge difference in regret amount between no observation capability within the attacking slot and one observation capability.
%In order to better observe the results we have plotted the figures for $T$ from $1$ to $2000$. 
%We can observe the following from the figures.
%\begin{itemize}
%\item In both figures regardless of the PUs activity type, the occurred regrets are below the regret upper bound achieved from the theoretical analysis. 
% As a result they comply with the theoretical analysis.
%\item In both figures, as the number of channels increases the regret increases as well as is expected based on the theoretical analysis from Corollary (3). 
%\item As the number of channels increases, the regret does not increase linearly with it. Instead, the increment in the regret becomes marginal which complies with the theoretical analysis. Based on Corollary (3) the regret is proportional to $(\sqrt {K \ln K})$.
%\end{itemize}

\begin{figure}
						  \centering
								\includegraphics[scale=.48]{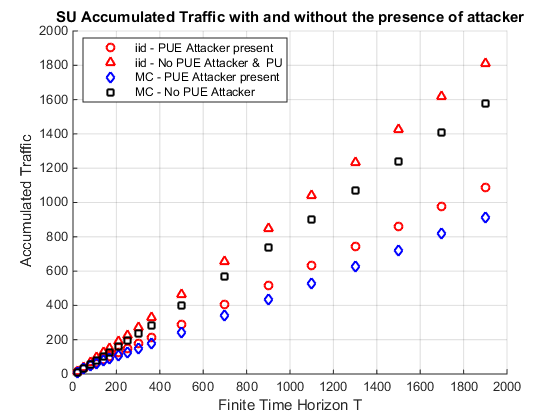}
							\caption{Accumulated Traffic of an SU}
							\label{AccSUfig}
\end{figure}	

%% file: 14Simltion3.tex
\subsection{Impact of the number of observations in each time slot}	
\label{mObservations}

We consider PROLA algorithm with $K=40$ channels in the network.
The number of observing channels, $m$, varies from 1 to 35.
Figure~\ref{Simluton}(h) and (i) show the performance of the PROLA for $m=1, 3, 8, 18, 35$ when the PUs follow i.i.d. distribution and Markovian chain on the channels, respectively. 
We can observe the following from the simulation results.

\begin{itemize}

\item As the observation capability of the attacker increases, it achieves a lower regret. 
This observation complies with the Corollary 3 provided in Section \ref{MultiactionObsOPT} based on which 
we expect smaller constant factor as the observation capability increases.

\item In the beginning, even adding a couple of more observing channels (from $m=1$ to $m=3$), the regret decreases dramatically.
The decrement in the regret becomes marginal as the number of observing channels becomes sufficiently large (e.g., from $m=18$ to $m=35$).
This observation implies that, in order to achieve a good attacking performance (smaller constant factor in regret upper bound), 
the attacker does not need to be equipped with high observation capability. % as is the case with Hedge \cite{PACB1995}.
In the simulation, when the number of observing channels ($m=10$) is $\frac{1}{4}$ of the number of all channels ($K=40$), the regret is approaching to the optimal.

\end{itemize}

%% file: 15Simltion4.tex
\subsection{Accumulated Traffic of SU with and without attacker}	

We set $K=10$ and measure the accumulated traffic achieved by the SU with and without the presence of the attacker. 
The attacker employs the PROLA algorithm.
Figure~\ref{AccSUfig} shows that the accumulated traffic of the SU is largely decreased when there is a PUE attacker in the network for both types of PU activities.

%% file: 16Conclusions.tex
\section{Conclusions and Future Work}
\label{Conclusions}

In this paper, we studied the optimal online learning algorithms that can be applied by a PUE attacker 
without any prior knowledge of the primary user activity characteristics and secondary user channel access policies.
We formulated the PUE attack as an online learning problem.
We identified the uniqueness of PUE attack that even though a PUE attacker cannot observe the reward on the attacking channel, but is able to observe the reward on another channel.
%We utilize the characteristic of a PUE attacker that even though the PUE attacker cannot observe the reward on the attacking channel, but is able to observe at least one other channel if the time slot duration is long enough.
We proposed a novel online learning strategy called POLA for the attacker with no observation capability in the attacking time slot. 
In other words, this algorithm is suitable when simultaneous attack and observation is not possible within the same time slot.
POLA dynamically decides between attack and observation, then chooses a channel for the decision made.
POLA achieves a regret in the order of $ \tilde{\Theta}(\sqrt[3]{T^2})$.
We showed POLA's optimality by matching its regret upper and lower bounds.
We proposed another novel online learning algorithm called PROLA for an attacker with at least one observation capabilities. 
For such an attacker, the attack period is followed by an observation period during the same time slot.
PROLA introduces a new theoretical framework under which the agent achieves an optimal regret in the order of $ \tilde{\Theta}(\sqrt{T})$.
%The optimality of the PROLA is realized as its upper bound regret matches its lower bound regret. 
One important conclusion of our study is that with no observation at all in the attacking slot in the POLA case, 
the attacker loses on the regret order, and with the observation of at least one channel in the PROLA case, 
there is a significant improvement on the performance of the attacker. 
This is in contrast to the case where increasing the number of observations 
from one to $m \geq 2$, does not make that much difference, only improving the constant factor.
Though, this observation can be utilized to study the approximate number of observations required to get the minimum constant factor.
The attacker's regret upper bound has a dependency on the number of observations $m$ as $\sqrt{1/m}$.
That is, the regret decreases overall for an attacker with higher observation capability (larger $m$).
However, when the number of observing channels is small, the regret decreases more if we add a few more observing channels.
While, the decreased regret will become marginal when more observing channels are added.
This finding implies that an attacker may only need a small number of observing channels to achieve a good constant factor.
The regret upper bound also is proportional to $\sqrt{K \ln K}$ which means the regret increases when there are more channels in the network.
The proposed optimal learning algorithm, PROLA, also advances the study of online learning algorithms.
It deals with the situation where a learning agent cannot observe the reward on the action that is taken but can partially observe the reward of other actions.
Before our work, the regret upper bound is proved to be in the order of $ \tilde {O} (\sqrt[3]{T^2})$.
Our algorithm achieves a tighter regret bound of $\tilde{O}(\sqrt{T})$ by randomizing the observing actions which introduces the concept of time-variable random feedback graphs.
We show this algorithm's optimality by deriving its regret lower bound which matches with its upper bound.

As for future work, we believe that our work can serve as a stepping stone to study many other problems.
How to deal with multiple attackers will be interesting, especially when the attackers perform in a distributed manner.
One other interesting direction is to study the equilibrium between the PUE attacker(s) and secondary user(s) when both of them employ learning based algorithms.
Integrating non-perfect spectrum sensing and the ability of PUE attack detection into our model will also be interesting especially that the PUE attacker may interfere with the PU during spectrum sensing period which can make it detectable by the PU.
From theoretical point of view, 
%it shows that introducing randomization  make a significant improvement and deserves more study.
% Moreover, 
our result shows only one possible case that the feedback graph despite being partially observable, achieves a tighter bound. 
A particular reward structure may allow for $\tilde{O}(\sqrt{T})$ regret even in partially observable graphs.

%% file: 17Appendix.tex
\label{APPNDX}

\subsection{Proof of Theorem 2}
\begin{proof}[Proof of Theorem 2]

We sketch the proof as follows. 
The problem of the PUE attacker with no observation capability within the attacking time slot can be modeled as a feedback graph.
Feedback graphs are introduced in \cite{NANC2015} based on which the observations governing the actions are modeled as graphs.
More specifically, in a feedback graph, the nodes represent the actions and the edges connecting them show the observations made while taking a specific action.
Fig. \ref{FigPOLAFdBck} shows how our problem can be modeled as a feedback graph.

In this figure, nodes $1$ to $K$ represent the overall number of actions. 
There are $K$ more actions however in this figure. 
These $K$ extra actions are required to be able to model our problem with the idea of feedback graphs. 
Actions $K+1$ up to $2K$ represent the observations; 
i.e., any time the agent decides to make an observation, it is modeled as an extra action.
There are $K$ channels and we model any observation on each channel with a new action which adds up to $2K$ actions overall.
The agent gains a reward of zero if it chooses the observation action since it is not a real action.
Instead, it makes an observation on the potential reward of its associated real action.
%Based on this model, if the attacker chooses a node belonging to observation actions say $i+K$, it observes the reward on the observation action itslef 
%which is zero since the attacker gains nothing by taking this action and also observes the reward on the real action $i$.

So, this problem can be modeled as a feedback graph and it in fact turns out to be a partially observable graph \cite{NANC2015}.
In \cite{NANC2015}, it is proved that the regret lower bound for partially observable graphs is $\Omega(\nu \sqrt[3]{KT^2})$
which completes the proof.
\end{proof}

%This means that our regret upper bound matches the regret lower bound which shows the optimality of our proposed algorithm, POLA.

\subsection{Proof of Theorem 4}

The proof here is similar to the lower bound analysis of EXP3 given in proof of Theorem 5.1, Appendix, in \cite{PACB2003}.
We mention the differences here. 
For this analysis the problem setup is exactly the same as the one in \cite{PACB2003}.
The notations and definitions are also the same as given in the first part of the analysis in \cite{PACB2003}. 
Below we bring some of the important notations and definitions used in the analysis here to keep the clarity of our analysis and to make our analysis self contained; 
however, the reader is referred to \cite{PACB2003} for more details on the definitions of notations.

\begin{figure}
\centering
\includegraphics[width=.28\textwidth]{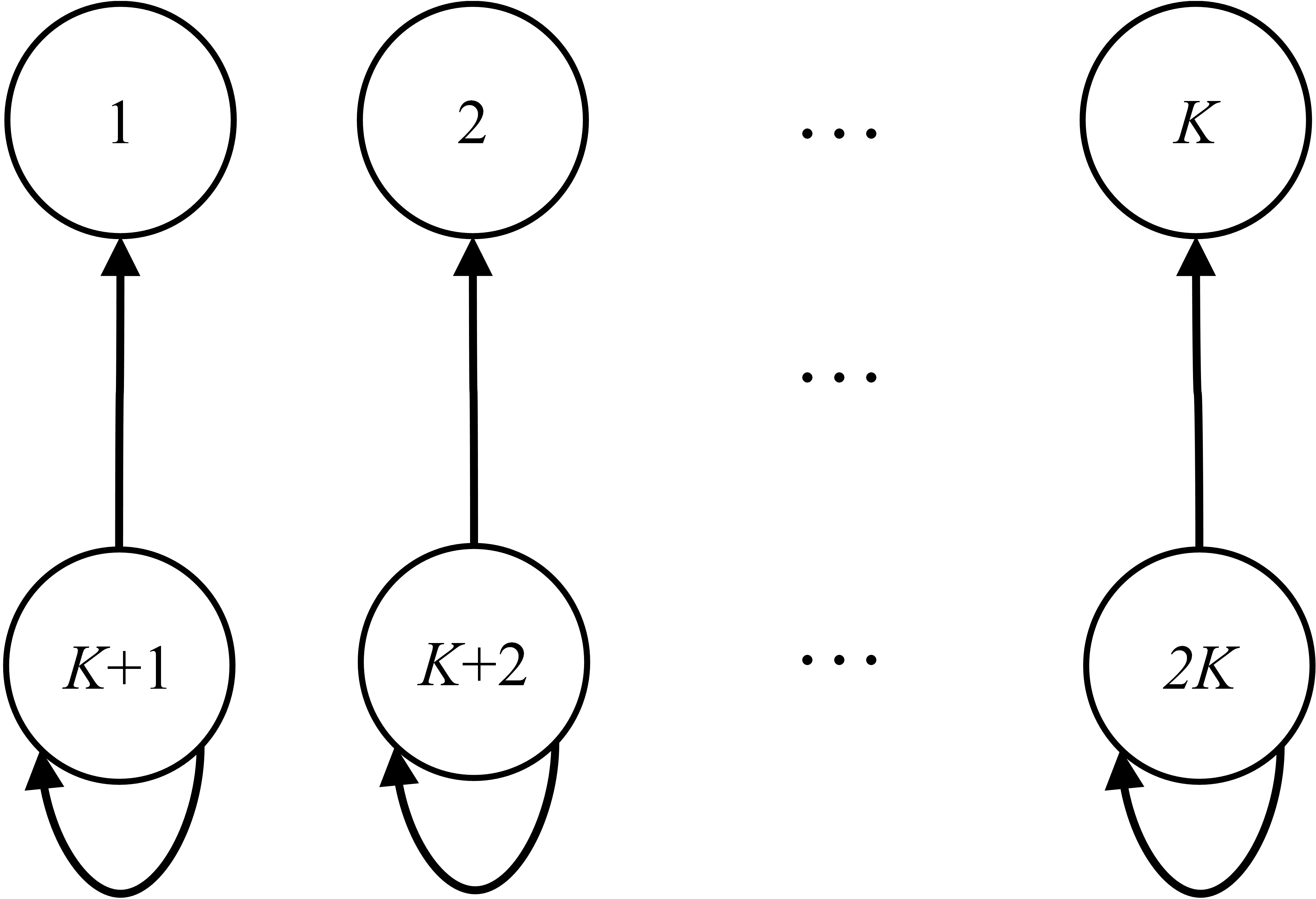}
\caption{Modeling the Attacker with no observation capability within the attacking slot with a feedback graph}
\label{FigPOLAFdBck}
\end{figure}

The reward distribution on the actions are defined as follows.
One action is chosen uniformly at random to be the good action.
Below is the reward distribution on the good action, $i$ for all $t= 1, ..., T$,
\begin{align}
x_t(i) =
\begin{cases}
1, & 1/2 + \epsilon \\
0, & o.w.,
\end{cases}
\end{align}

where $\epsilon \in (0, 1/2]$. 
The reward distribution on all other actions is defined to be one or zero with equal probabilities.
$\bold{P}_\ast \left\{ \cdot \right\}$ is used to denote the probability w.r.t. this random choice of rewards to play.
$\bold{P}_i \left\{   \cdot \right\}$ represents the probability conditioned on $i$ being the good action.
Also, $\bold{P}_{unif} \left\{ \cdot  \right\}$ shows the probability with respect to uniformly random choice of rewards on all actions.
$\bold{O}$ notation is also used to denote the probability w.r.t. observation.
Analogous observation probability $\bold{O}_\ast \left\{ \cdot \right\}$, $\bold{O}_i \left\{ \cdot \right\}$, $\bold{O}_{unif} \left\{ \cdot\right\}$
and expectation notations, $\bold{E}_\ast [ \cdot]$, $\bold{E}_i [\cdot]$, $\bold{E}_{unif} [\cdot]$ are used.

The agent's access/attack policy is denoted by $A$. 
$r_t= x_{t}(i_t)$ is a random variable which its value shows the reward gained at time $t$.
$\bold{r}^t = <r_1, \ldots, r_t>$ is a sequence of rewards received till $t$ and $\bold{r} $ is the entire sequence of rewards.
$G_A = \sum\limits_{t=1}^T r_t $ is the gain of the agent and $G_{max} = \max\limits_j \sum\limits_{t=1}^T x_t(j)$.
The number of times action $i$ is chosen by A is a random variable denoted by $N_i$.

Lemma 1  
\emph{Let}$f: \left\{0,1\right\}^T \longrightarrow [0,M]$ \emph{be any function defined on reward sequences} $\bold{r}$.
\emph{Then for any action} $i$, \\

$\bold{E}_i[f(\bold{r})] \leq \bold{E}_{unif}[f(\bold{r})] + \frac{M}{2} \sqrt{-\bold{E}_{unif}[\bold{O}_i] \ln (1-4 \epsilon^2)}$.

\begin{proof}
\begin{align}
\label{equ501}
\bold{E}_i[f(\bold{r})] - \bold{E}_{unif}[f(\bold{r})] &= \sum\limits_{\bold{r}} f(\bold{r}) (\bold{O}_i\left\{\bold{r}\right\} - \bold{O}_{unif}\left\{\bold{r}\right\})  \nn
& \leq \sum\limits_{\bold{r}: \bold{O}_i\left\{\bold{r}\right\} \geq \bold{O}_{unif}\left\{\bold{r}\right\}}   \nn
& \quad\quad\quad\quad\quad f(\bold{r}) (\bold{O}_i\left\{\bold{r}\right\} - \bold{O}_{unif}\left\{\bold{r}\right\})   \nn
& \leq M \sum\limits_{\bold{r}: \bold{O}_i\left\{\bold{r}\right\} \geq \bold{O}_{unif}\left\{\bold{r}\right\}} \nn
& \quad\quad\quad\quad\quad\quad (\bold{O}_i\left\{\bold{r}\right\} - \bold{O}_{unif}\left\{\bold{r}\right\})  \nn
&= \frac{M}{2}  \left\| \bold{O}_i - \bold{O}_{unif}  \right\|_1.
\end{align}

We also know from \cite{TMJT1991,PACB2003} that, 

\begin{equation}
\label{equ502}
\left\| \bold{O}_{unif} - \bold{O}_i  \right\|_1^2 \leq (2 \ln 2) KL (\bold{O}_{unif} \left|\right| \bold{O}_i).
\end{equation}

From chain rule for relative entropy we derive the following,
\begin{align}
\label{equ503}
KL (\bold{O}_{unif} \left|\right| \bold{O}_i) &= \sum\limits_{t=1}^T KL(\bold{O}_{unif} \left\{r_t | \bold{r}^{t-1}\right\} \left|\right| \bold{O}_{i} \left\{r_t | \bold{r}^{t-1} \right\} ) \nn
& = \sum\limits_{t=1}^T (\bold{O}_{unif}\left\{i_t \neq i \right\} KL (\frac{1}{2} \left|\right|\frac{1}{2})) \nn
& \quad\quad + (\bold{O}_{unif}\left\{i_t = i \right\}   KL (\frac{1}{2} \left|\right|\frac{1}{2}+\epsilon))  \nn
&	=	\sum\limits_{t=1}^T		\bold{O}_{unif}\left\{i_t = i \right\} (-\frac{1}{2} \lg (1-4 \epsilon^2))	\nn
& = \bold{E}_{unif}[\bold{O}_i] (-\frac{1}{2} \lg (1-4 \epsilon^2)).
\end{align}

The lemma follows by combining (\ref{equ501}),(\ref{equ502}), and (\ref{equ503}).
\end{proof}

\begin{proof}[Proof of Theorem 4]

The rest of the analysis in this part is similar to the analysis in Theorem A.2 in \cite{PACB2003}, except that when we apply
lemma 1 to $N_i$, we reach the following inequality,

$\bold{E}_i[N_i] \leq \bold{E}_{unif}[N_i] + \frac{T}{2} \sqrt{-\bold{E}_{unif}[\bold{O}_i] \ln (1-4 \epsilon^2)}$.\\
where $\sum\limits_{i=1}^K \sqrt{\bold{E}_{unif}[\bold{O}_i]} \leq \sqrt{KT}$.
%since based on the observation policy the actions are chosen uniformly at random other than the attacked one. 
By considering the observation probability and making a little simplification the upper bound is achieved.
Following similar steps as in Theorem A.2 in \cite{PACB2003} for the rest of the analysis gives us the regret lower bound equal to $\omega(c\sqrt{KT}) $ which completes the proof.
\end{proof}